\documentclass[a4paper,twocolumn,11pt,accepted=2023-01-01]{quantumarticle}
\pdfoutput=1

\PassOptionsToPackage{compress}{natbib}

\usepackage[numbers,sort&compress]{natbib}
\usepackage{ifpdf}
\usepackage{amsmath} 
\usepackage{amssymb}
\usepackage{amsfonts}
\usepackage{braket}
\usepackage{color}
\usepackage{hyperref}
\usepackage{graphicx}
\usepackage{dcolumn}
\usepackage{bm}
\usepackage[normalem]{ulem}
\usepackage[utf8]{inputenc}
\usepackage[english]{babel}
\usepackage[T1]{fontenc}

\usepackage{stackengine}

\newcommand{\beq}{\begin{eqnarray}}
\newcommand{\eeq}{\end{eqnarray}}

\newcommand{\la}{\langle}
\newcommand{\ra}{\rangle}
\newcommand{\tr}{\text{Tr}}

\newcommand{\bfsigma}{\boldsymbol{\sigma}}
\newcommand{\MIE}{\text{MIE}}
\newcommand{\MIC}{\text{MIC}}
\newcommand{\MI}{\text{MI}}

\definecolor{Zcolour}{rgb}{0.992, 0.588, 0.22}
\definecolor{purple}{rgb}{0.5,0,0.5}
\definecolor{brown}{rgb}{0.6,0.2,0}
\definecolor{dkgreen}{rgb}{0,0.5,0}
\newcommand{\authorcomment}[3]{{ \color{#1} \footnotesize (\textsf{#2}) \textsf{\textsl{#3}} }}

\newcommand{\jacob}[1]{\authorcomment{blue}{JL}{#1}}

\usepackage{mathtools}

\newcommand{\mc}[1]{\mathcal{#1}}

\usepackage[utf8]{inputenc}
\usepackage[english]{babel}
\usepackage{amsthm}
\usepackage{dutchcal}

\newtheorem{theorem}{Theorem}[section]

\newtheorem{lemma}[theorem]{Lemma}
\newtheorem{definition}{Definition}[section]

\usepackage{stackengine}

\usepackage{comment}
\usepackage[caption=false]{subfig}

\begin{document}


\title{Probing sign structure using measurement-induced entanglement}
\author{Cheng-Ju Lin}
\thanks{These authors contributed equally.}
\email{jacob0425@gmail.com}
\affiliation{Perimeter Institute for Theoretical Physics, Waterloo, Ontario N2L 2Y5,  Canada}
\author{Weicheng Ye}
\thanks{These authors contributed equally.}
\affiliation{Perimeter Institute for Theoretical Physics, Waterloo, Ontario  N2L 2Y5, Canada}
\affiliation{Department of Physics and Astronomy, University of Waterloo, Waterloo, Ontario N2L 3G1, Canada}
\author{Yijian Zou}
\thanks{These authors contributed equally.}
\affiliation{Stanford Institute for Theoretical Physics, Stanford University, Palo Alto, CA 94305, USA}
\author{Shengqi Sang}
\affiliation{Perimeter Institute for Theoretical Physics, Waterloo, Ontario N2L 2Y5, Canada}
\affiliation{Department of Physics and Astronomy, University of Waterloo, Waterloo, Ontario N2L 3G1, Canada}
\author{Timothy H. Hsieh}
\email{quantim5@gmail.com}
\affiliation{Perimeter Institute for Theoretical Physics, Waterloo, Ontario N2L 2Y5, Canada}


\begin{abstract}
The sign structure of quantum states is closely connected to quantum phases of matter, yet detecting such fine-grained properties of amplitudes is subtle.  Here we employ as a diagnostic measurement-induced entanglement (MIE)-- the average entanglement generated between two parties after measuring the rest of the system.  We propose that for a sign-free state, the MIE upon measuring in the sign-free basis decays no slower than correlations in the state before measurement.  Concretely, we prove that MIE is upper bounded by mutual information for sign-free stabilizer states (essentially CSS codes), which establishes a bound between scaling dimensions of conformal field theories describing measurement-induced critical points in stabilizer systems.  We also show that for sign-free qubit wavefunctions, MIE between two qubits is upper bounded by a simple two-point correlation function, and we verify our proposal in several critical ground states of one-dimensional systems, including the transverse field and tri-critical Ising models.  In contrast, for states with sign structure, such bounds can be violated, as we illustrate in critical hybrid circuits involving both Haar or Clifford random unitaries and measurements, and gapless symmetry-protected topological states.          
\end{abstract}

\maketitle

\section{Introduction}

The sign problem is an infamous bottleneck \cite{PhysRevB.41.9301,PhysRevLett.94.170201} for quantum Monte Carlo simulations, whose efficiency requires non-positive off-diagonal matrix elements of a Hamiltonian with respect to a local basis~\cite{10.1143/PTP.56.1454,gubernatis_kawashima_werner_2016,Bravyi2006complexity}.  However, recent works have shown that the sign problem is not all negative news due to the computational complexity; rather, there are positively fascinating connections to quantum phases of matter.  In particular, many topologically ordered phases of matter including the double semion phase have an ``intrinsic sign problem'' \cite{H16,RK17,SGR20}, defined by the lack of any local basis in which the Hamiltonian is sign-free.

Likewise, the sign structure of states/wavefunctions has also proven to be physically significant.  A wavefunction has an intrinsic sign problem if there is no finite depth circuit which can transform it into a non-negative wavefunction~\cite{H16, torlaiWavefunction2020}.  In Ref.~\cite{Ellison2010}, it was shown that many symmetry-protected topological (SPT) phases have a symmetry-protected sign problem, which no symmetric finite depth circuit can remove.  Moreover, the sign structure of highly excited and random states is intimately connected to their volume law entanglement scaling \cite{GF15}.

Given these physical aspects of sign structure in wavefunctions, how does one detect it without examining exponentially many amplitudes?  Sign-free states are still capable of exhibiting a variety of quantum phenomena including quantum criticality and topological order, so how do they differ in physical properties from sign-structured states?  In Ref.~\cite{H16}, Hastings found a property specific to sign-free states in one dimension and with zero correlation length.  For such states, the amount of entanglement that can be generated between two parties if one measures the rest of the system in the sign-free basis--a diagnostic we will call measurement-induced entanglement (MIE)-- is super-polynomially small in the distance between the two parties.  A rough picture for this result is that such states are approximately ``coherent Gibbs states'', whose amplitudes are Boltzmann weights of a local classical Hamiltonian, and thus measuring or fixing a configuration outside of two parties will lead to a factorized state between the two parties.  Quantities akin to MIE have also been utilized in \cite{Popp2005,ESBD12prl,Grover_2014,Rajabpour_2016,PhysRevB.101.115131,PhysRevB.95.045111,Vijay2020,li2021conformal}.

In this work, we use MIE to probe the sign structure of wavefunctions beyond zero correlation, with a focus on quantum critical states in which the correlation length is effectively infinite.  We find that in a variety of sign-free critical and off-critical systems, MIE must decay at least as quickly as the slowest decaying correlation function in the wavefunction before any measurement; colloquially, measuring sign-free states in the sign-free basis cannot generate significantly more correlations than those that already exist. 

Specifically, we show that for sign-free stabilizer states (essentially Calderbank-Shor-Steane or CSS codes), MIE is upper bounded by mutual information.  We apply this bound to the recently discovered \textit{measurement induced phase transition} (MIPT), which emerge from the competition between two types of non-commuting measurements \cite{Nahum2020, Sang2021,barkeshli,ippoliti,lang2020entanglement} or the competition between random unitary gates and measurements \cite{li2019measurement,nahum2018hybrid,nandkishore2018hybrid,hybridreview}.  We show that for MIPT within the stabilizer formalism, the sign-free condition implies a bound between the lowest scaling dimensions of certain boundary condition changing operators of the associated conformal field theory. 
We illustrate this bound in two measurement-only quantum circuits, one of which is a new critical model involving four-body Majorana measurements.  In contrast, critical states with sign structure, which arise from dynamics involving both measurements and unitary gates, can violate the above bound.  We illustrate this for both hybrid circuits involving Clifford random gates and those involving Haar random gates.  As a byproduct, we show that the scaling of MIE can distinguish between these two (Haar vs. Clifford) critical universality classes, while mutual information cannot.

We also show that for sign-free qubit wavefunctions, MIE between two qubits is upper bounded by a two-point correlation function, and we apply this bound to critical ground states of one-dimensional local Hamiltonians.  We illustrate the bound for the transverse field and tri-critical Ising models, and we show that the critical three-state Potts model obeys a similar bound despite its different local Hilbert space dimension.  Remarkably, we also find evidence that the MIE is a function of the conformal cross-ratio only.   Critical states with sign-structure can violate these bounds, as we illustrate for gapless SPTs \cite{scaffidi, verresen} in one dimension.  Finally, we conclude with several open questions regarding intrinsic sign problems for critical states. 

\section{Setup}

\begin{figure}
    \centering
    \includegraphics[width=1\columnwidth]{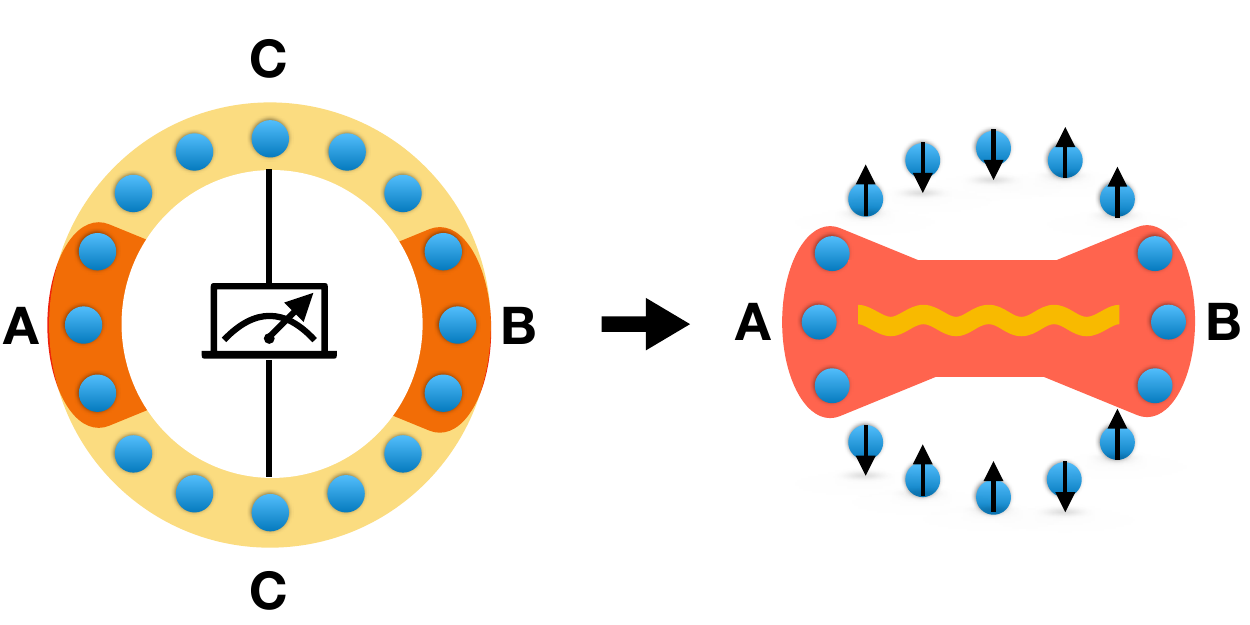}
    \caption{Measurement-induced entanglement.  We find that for sign-free states in many physical contexts, measuring $C$ in the sign-free basis cannot generate more entanglement between $A,B$ than correlations existing before measurement.  In contrast, the latter is achievable by states with sign structure.}
    \label{fig:MIEillustration}
\end{figure}

We now define the central quantity in this work: \textit{measurement induced entanglement} (MIE). 
We consider a system with $N$ sites, each site with local Hilbert space dimension $d$ (we will mostly focus on the case of qubits where $d=2$). 
After partitioning the system into three parts $A$, $B$, and $C$, we ask how much entanglement can be generated on average between $A$ and $B$ by measuring $C$ in a given basis, following Ref.~\cite{H16,Popp2005}.  We will be interested in partitions where $A,B$ are small intervals far apart and $C$ is everything else, see Fig.~\ref{fig:MIEillustration}.  
Specifically, for a wavefunction 
\beq
 |\psi\ra = \sum_{a,b,c} \psi_{abc}|abc\ra~,
\eeq
the probabilities and the resulting wavefunctions for obtaining a given outcome $c$ after measuring $C$ in the given computational basis are $p_c = \sum_{a,b}|\psi_{abc}|^2$ and $|\chi_c\ra = \sum_{a,b} \chi_{ab}|ab\ra$, respectively, where $\chi_{ab} = \psi_{abc}/\sqrt{p_c}$.
The MIE is defined as 
\beq
    \MIE(A:B) \equiv \sum_c p_c S_A(|\chi_c\ra)~,
\eeq
where the entanglement between $A$ and $B$ in the state $|\chi_c\ra$ is measured by  
von Neumann entropy $S_A(|\chi_c\ra)  = - \tr[\rho_{A,c} \ln \rho_{A,c}]$ and $\rho_{A,c} = \tr_B[|\chi_c\ra\la\chi_c|]$ is the reduced density matrix on $A$.

We will also use $\MIE_X$ to denote the MIE measured in the $X$-basis, and analogously for other bases.  If none is explicitly stated, the $Z$ (computational) basis is implied.
Note that MIE is very similar to the localizable entanglement proposed in Ref.~\cite{Popp2005}, but here we do not consider optimizing over the measurement basis to maximize the induced entanglement.

\section{Sign structure in Stabilizer states}

We begin by analyzing the consequences of the sign-free condition for stabilizer states, which are eigenstates of commuting Pauli operator strings called stabilizers.  We find that a stabilizer state that is sign-free in the $Z$-basis has a special structure: its stabilizers can always be chosen to consist of strings of only Pauli-$X$ operators or only Pauli-$Z$ operators --- such states are also known as Calderbank-Shor-Steane (CSS) codes \cite{Calderbank1996,Steane1996}.
See Appendix \ref{app:sign-free stablizer} for a detailed derivation.

\subsection{Bound on MIE for sign-free stabilizer states}

We show that for a stabilizer state that is sign-free in the computational basis, the measurement-induced entanglement is upper bounded by the mutual information
\beq\label{eq:MIE_stab_bound}
\MIE(A:B) \leq  \MI(A:B)
\eeq
where the mutual information $\MI(A:B)=S_A+S_B-S_{AB}$. 
In other words, one needs existing classical or quantum correlations to be able to generate entanglement after measuring $C$ in the $Z$-basis if the stabilizer state is sign-free.

To prove this bound, we build on the structure theorem for stabilizers \cite{Bravyi2006}, which states that for a stabilizer state $|\psi\rangle$ supported on $A,B,C$, there exist local unitaries $U_A, U_B, U_C$ which deconstruct $|\psi\rangle$ into a direct product of three classes of simple states:
{\small
\beq
    \label{eq:structural_thm_elementary}
   U_A U_B U_C \ket{\psi} 
    &&= 
    \ket{\mathrm{GHZ}}_{ABC}^{\otimes g_{ABC}} \nonumber\\
    &&\otimes
    \ket{\mathrm{EPR}}_{AB}^{\otimes e_{AB}} \otimes
    \ket{\mathrm{EPR}}_{BC}^{\otimes e_{BC}} \otimes
    \ket{\mathrm{EPR}}_{CA}^{\otimes e_{CA}} \nonumber
    \\ 
    &&\otimes \ket{0}_A^{\otimes s_{A}}
    \otimes \ket{0}_B^{\otimes s_{B}}
    \otimes \ket{0}_C^{\otimes s_{C}}, 
\eeq
}%
where $|\mathrm{GHZ}\rangle = \frac{1}{\sqrt{2}}(|000\rangle + |111\rangle)$ and $|\mathrm{EPR}\rangle = \frac{1}{\sqrt{2}}(|00\rangle + |11\rangle)$.

For our purposes, we need a finer structure theorem specific to the sign-free property. In Appendix \ref{app:structure_thm}, we show that for a sign-free stabilizer $|\overline{\psi}\rangle$, there exist local unitaries $\overline{U}_A, \overline{U}_B, \overline{U}_C$ that preserve the sign-free structure (that map strings of Pauli $X$ ($Z$) operators to strings of Pauli $X$ ($Z$) operators) such that  
{\small
\beq
    \label{eq:structural_thm_CSS}
   \overline{U}_A \overline{U}_B \overline{U}_C \ket{\overline{\psi}} 
   &&=
    \ket{\mathrm{GHZ}}_{ABC}^{\otimes g_{ABC}} \otimes \ket{\mathrm{GHZ'}}_{ABC}^{\otimes g'_{ABC}} \nonumber \\
    &&\otimes
    \ket{\mathrm{EPR}}_{AB}^{\otimes e_{AB}} \otimes
    \ket{\mathrm{EPR}}_{BC}^{\otimes e_{BC}} \otimes
    \ket{\mathrm{EPR}}_{CA}^{\otimes e_{CA}} \nonumber
    \\ 
    &&\otimes \ket{0}_A^{\otimes s_{A}}
    \otimes \ket{0}_B^{\otimes s_{B}}
    \otimes \ket{0}_C^{\otimes s_{C}} \nonumber
    \\
    &&\otimes \ket{+}_A^{\otimes s'_{A}}
    \otimes \ket{+}_B^{\otimes s'_{B}}
    \otimes \ket{+}_C^{\otimes s'_{C}}, 
\eeq
}%
where $|\mathrm{GHZ}'\rangle = \frac{1}{\sqrt{2}}(|+++\rangle + |---\rangle)$ and $\ket{\pm}$ are eigenstates of the $X$ operator: $X\ket{\pm}=\pm \ket{\pm}$.
The fact that $\overline{U}_C$ is sign-free preserving implies that it commutes with any projective measurement on $C$ in the $Z$ basis, and therefore the calculation of MIE reduces to evaluating it for the multiple simple states above.  It is straightforward to check that the only states which give non-zero contribution to either MIE or MI are $|\mathrm{GHZ}\rangle_{ABC}$, for which $\MIE(A:B)=0$, $\MI(A:B)=1$, $|\mathrm{GHZ}'\rangle_{ABC}$, for which $\MIE(A:B)=\MI(A:B)=1$, and  $|\mathrm{EPR}\rangle_{AB}$, for which $\MIE(A:B) =1$ and $\MI(A:B)=2$.  This proves the bound $\MIE(A:B) \leq  \MI(A:B)$ for sign-free stabilizers.  We present an alternative proof in Appendix~\ref{app:alternative_derivation}.  The sign-free stabilizer states which saturate the bound consist (up to local unitaries) of $|\mathrm{GHZ}'\rangle_{ABC}$ states and states which do not contribute to either $\MIE(A:B)$ or $\MI(A:B)$: $|\mathrm{EPR}\rangle_{AC}$, $|\mathrm{EPR}\rangle_{BC}$, and direct product states.

\subsection{Violation of the bound due to sign structure}

Here we provide an example of a stabilizer state with sign structure, which can have nonzero $\MIE(a:b)$ despite having zero $\MI(a:b)$ for two qubits $a,b$.
Consider the one-dimensional cluster state with periodic boundary conditions, stabilized by the set $\{X_i Z_{i+1} X_{i+2}\}$ for every site $i$. This state has zero correlation length; the mutual information between distant qubits is zero.  However, MIE from measuring in the $Z$ basis is infinitely long-ranged due the string order parameter of the cluster state.  The string is obtained by multiplying every other stabilizer to cancel intermediate $X$s, yielding an operator $X_a Z_{a+1}...Z_{b-1} X_b$ with only $Z$s or identities in the middle of the string.  Being a product of stabilizers, such a string operator is also a stabilizer.  Similarly, the global Ising operator $\prod Z$ is also a stabilizer.  Hence, measuring in the $Z$-basis all but two qubits at $a,b$ with even distance results in a residual state on $a,b$ stabilized by $X_a X_b, Z_a Z_b$, which is an EPR pair with one bit of entanglement.  As $a,b$ can be arbitrarily far apart, the $\MIE$ does not decay.
This is the ``quantum wire'' property of 1d SPTs \cite{ESBD12prl}, a manifestation of a symmetry-protected sign problem in the symmetry charge ($Z$) basis \cite{Ellison2010}.

\section{MIE in measurement-driven critical states}

In this section, we apply our bound to derive consequences for critical stabilizer states arising from quantum dynamics involving both unitary evolution and projective measurements.  In particular, we consider random circuits in which each local operation is either a Clifford unitary (which maps Pauli strings to Pauli strings in the Heisenberg representation) or Pauli measurements.  Such circuits acting on a product state $\otimes |+\rangle$ preserve the stabilizer nature of the state at all times.  Remarkably, the ensemble of steady states produced by many different realizations of the random circuit can exhibit criticality described by a conformal field theory (CFT)~\cite{li2021conformal,jian2020}.

In more detail, after averaging over the randomness in the dynamics, entanglement entropy of a subregion of the steady state is mapped to the free energy cost of imposing different boundary conditions on the subregion and its complement \cite{rtn, PhysRevB.100.134203,jian2020,PhysRevB.101.104301}. In a CFT, such a change of boundary conditions can be achieved by inserting boundary condition changing operators. Thus, entanglement entropy of an interval can be mapped to two-point correlation functions of boundary operators, and entanglement entropy of two disjoint intervals $A=[x_1, x_2]$ and $B=[x_3, x_4]$ can be mapped to four-point correlation functions of boundary operators. 

As a direct consequence of the conformal symmetry, a four-point correlation function is a function only
of the \emph{cross ratio} $\eta$ \cite{belavin_infinite_1984}, which is defined as $\eta \equiv {w_{12}w_{34}}{w_{13}^{-1}w_{24}^{-1}}$, where $w_{ij}$ is the chord distance, i.e., $w_{ij}=\sin(\pi (x_i-x_j)/L)$. Importantly, the cross ratio $\eta$ is invariant under the conformal transformation. Therefore, we can extract universal exponents of the CFT by understanding the relationship between entanglement entropy of two disjoint intervals and $\eta$. Specifically, at small $\eta$ (when A and B are small and distant), MI scales as $\eta^{h_{\mathrm{MI}}}$, where $h_{\mathrm{MI}}$ is a universal exponent given by the scaling dimension of a boundary condition changing operator. 

Ref.~\cite{li2021conformal} found that MIE can also be mapped to four-point correlation functions of a different boundary condition changing operator. Thus MIE is also a function of the cross ratio, and the exponent $h_{\mathrm{MIE}}$ in $\MIE(A:B) \sim \eta^{h_{\mathrm{MIE}}}$ is also universal and given by a different scaling dimension.  To avoid clutter, we have suppressed the measurement basis for MIE but will detail it in the following specific examples. 

It follows from our bound that given any critical ensemble consisting of sign-free stablizer states, for any fixed $\eta$, the average of MIE must be no greater than the average of MI. 
More specifically, we assume the hybrid dynamics generates an ensemble of steady states $\rho=\sum_{n}p_n |\psi_n\rangle \langle \psi_n|$, where $|\psi_n\ra$ is a stabilizer state. 
If such a stabilizer state is sign-free, we would have $\MIE_n(A:B) \leq \MI_n(A:B)$, where the subscript $n$ denotes the associated quantity for $|\psi_n\ra$. 
Somewhat abusing the notation, we therefore have the inequality for the ensemble-averaged MIE and MI as $\MIE(A:B) \equiv \sum_n p_n\MIE_n(A:B) \leq \MI(A:B) \equiv \sum_n p_n\MI_n(A:B)$.
Since $\MIE(A:B) \sim \eta^{h_{\mathrm{MIE}}}$ and $\MI(A:B) \sim h_{\mathrm{MI}}$, the associated conformal field theory must have
\beq
h_{\mathrm{MIE}}\geq h_{\mathrm{MI}} \label{eq:CFTbound}
\eeq
for any measurement basis in which the states are sign-free.  It is remarkable that sign structure--naively, a UV property of the full many-body wavefunction--places constraints on the IR physics in this context.

\subsection{X-ZZ measurement-only critical point}\label{subsec:XXZ measurement}

\begin{figure}
    \centering
    \subfloat[]
    {\includegraphics[width=0.24\textwidth]{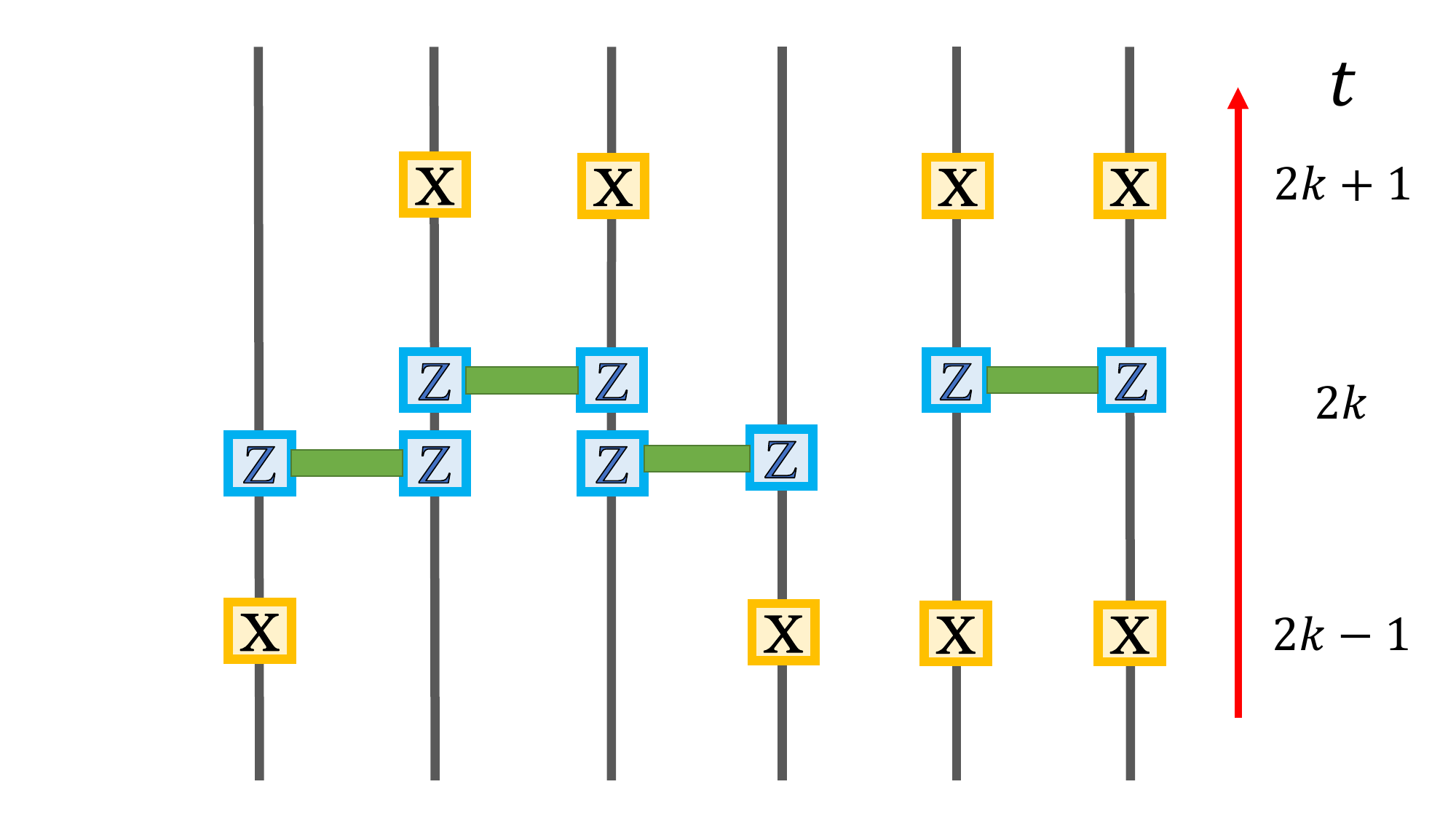}}
    \subfloat[]
    {\includegraphics[width=0.24\textwidth]{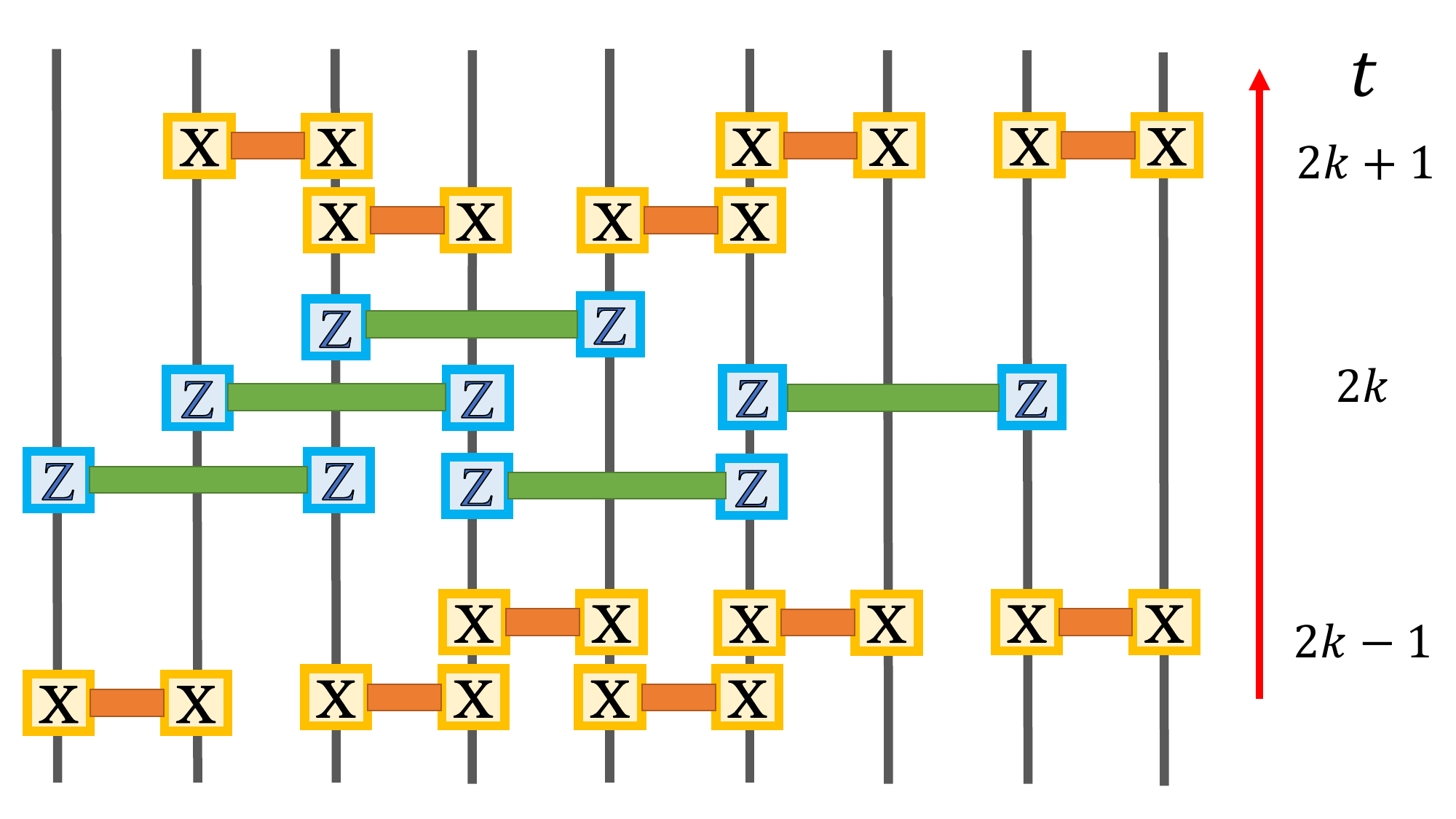}}
    \\
    \subfloat[]
    {\includegraphics[width=0.24\textwidth]{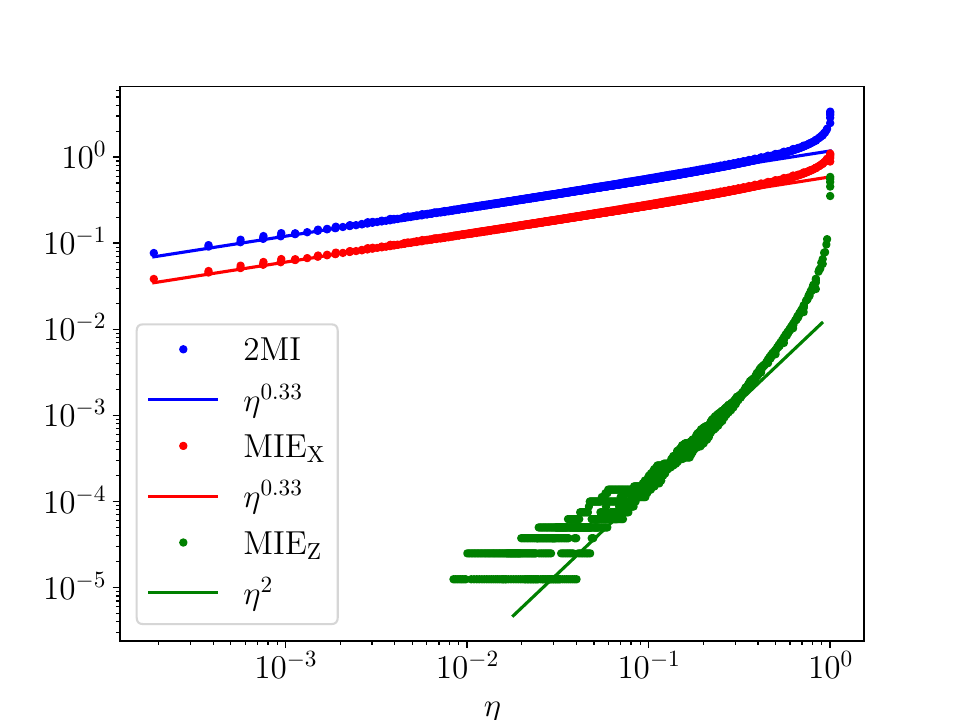}}
    \subfloat[]
    {\includegraphics[width=0.24\textwidth]{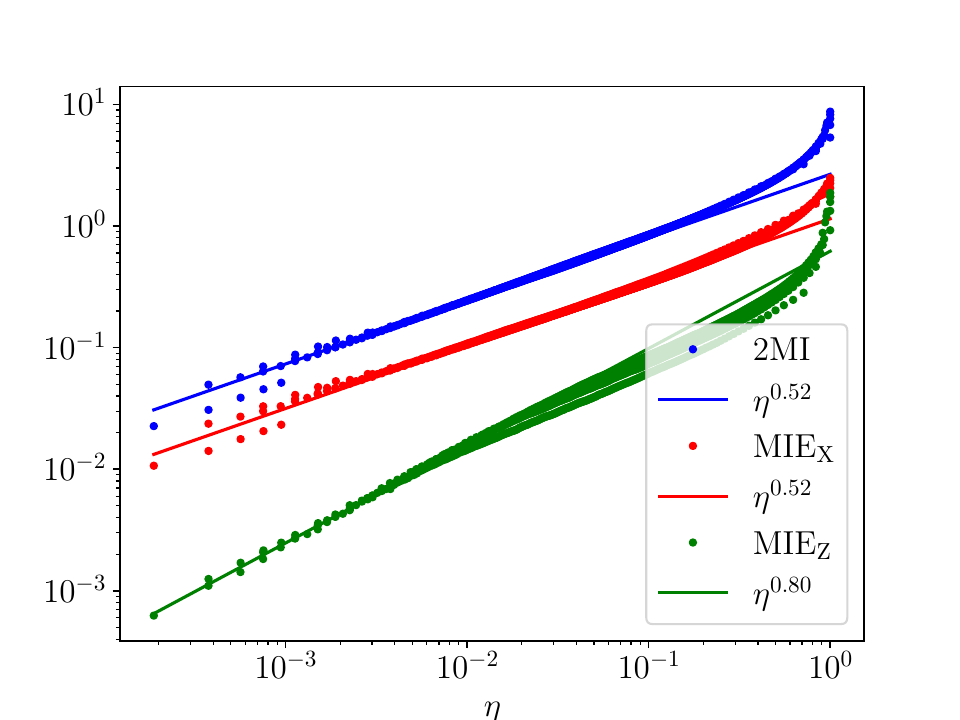}}
    \caption{Measurement-only dynamics and data at two critical points.  (a) Single-site $X$ and nearest-neighbor $ZZ$ measurements are performed with equal probability $p=1/2$.  (b) Nearest-neighbor $XX$ and next-nearest-neighbor $ZIZ$ measurements are performed with equal probability $p=1/2$.  (c)(d) MI and MIE for the steady states of the circuits in (a) and (b). For both cases, we have system size $L=512$ and average over 80000 samples for a single data point.}
    \label{fig:measurementonly}
\end{figure}

As the simplest example exhibiting this bound, we start with the critical point arising from competing $X$ and $ZZ$ measurements, both of which preserve sign-free (or CSS) structure.
This dynamics has been studied in Refs. \cite{ Nahum2020, Sang2021,lang2020entanglement,Sang2021b}, and we summarize it in the following.
At every odd time step, we perform $X$ measurement on each pair of neighbouring qubits $(i, i+1)$ with probability $p$, and at every even time step, we measure each qubit in the $ZZ$ basis with probability $1-p$, as illustrated in Fig.~\ref{fig:measurementonly}(a). 

Refs.~\cite{lang2020entanglement, Sang2021b} show that this $X$-$ZZ$ measurement dynamics, viewed as a (1+1)d quantum circuit, can be faithfully mapped to a bond percolation problem on the square lattice, where the presence(absence) of $X$($ZZ$) measurements corresponds to the broken vertical(horizontal) bonds. 
Furthermore, one can understand the wavefunction at any time as a tensor product of several multi-qubit GHZ states:
\begin{equation}\label{eq:zzx_wavefunction}
    \ket{\phi}=\ket{\text{GHZ}}_{\Sigma_1}\otimes...\otimes \ket{\text{GHZ}}_{\Sigma_k}~,
\end{equation}
where $\{\Sigma_1,...,\Sigma_k\}$ is a partition of all the qubits that is specified by the circuit's ``history'': two qubits belong to the same cluster (same $\Sigma$) if and only if they are connected by some percolating cluster in the 2d layout of the circuit. 
Note that this state is clearly sign-free in the $X$ or $Z$-basis.
When $p=1/2$, the $X$-$ZZ$ measurement dynamics undergoes a phase transition described by the 2d percolation CFT.

To compare MI and MIE between regions $A,B$ at this critical point, we define the following two sets:
{\small
\begin{align}
    s_1 &= \{\Sigma_i: \Sigma_i\cap A\neq \varnothing, \Sigma_i\cap B\neq \varnothing, \Sigma_i \cap \overline{A\cup B} = \varnothing\}~, \notag \\ 
    s_2 &= \{\Sigma_i: \Sigma_i\cap A\neq \varnothing, \Sigma_i\cap B\neq \varnothing, \Sigma_i \cap \overline{A\cup B} \neq \varnothing\}~. \notag
\end{align}
}%
i.e., $s_1 (s_2)$ consist of GHZ clusters supported exclusively (not exclusively) on $A,B$.   

Then we have
\begin{align}
    \MI(A:B)&=2|s_1|+|s_2|~,\\
    \MIE_X(A:B) &= |s_1|+|s_2|~,\\
    \MIE_Z(A:B) &= |s_1|~.
\end{align}
The percolation CFT predicts that both $|s_1|$ and $|s_2|$ are functions of the cross ratio $\eta$ only.
In particular, for small $\eta$, $|s_1| \propto \eta^{2}$ and $|s_2| \propto \eta^{1/3}$  \cite{lang2020entanglement, Sang2021b}. 
Therefore,
\begin{align}
    \MI(A:B)&\sim \eta^{1/3},\\
    \MIE_X(A:B) &\sim \eta^{1/3}, \\
    \MIE_Z(A:B) &\sim \eta^{2}~,
\end{align}
which is consistent with both our numerical simulation 
in Fig.~\ref{fig:measurementonly}(c) and our bound Eq.~(\ref{eq:CFTbound}).

\subsection{XX-ZIZ measurement-only critical point}

The steady states above are relatively simple as they are direct products of GHZ states.  Here, as a more non-trivial illustration of our bound, we consider the same random architecture as above, but we replace $X$ and $ZZ$ measurements by $XXI$ and $ZIZ$ measurements for three nearest-neighbor qubits, as illustrated in Fig.~\ref{fig:measurementonly}(b).  Again, the wavefunction generated by these measurements is always sign-free. 
This model is an interacting generalization of the dynamics in Sec.~\ref{subsec:XXZ measurement} : using Jordan-Wigner transformation, the $X$-$ZZ$ measurements map to two-Majorana parity measurements, while the $XX$-$ZIZ$ measurements map to four-Majorana parity measurements.

When $p=1/2$ (both measurements are done with equal probability), the dynamics is invariant, on average, under translation by a single Majorana mode and expected to yield a critical state \cite{Nahum2020,Hsieh2016}.  Indeed, we find evidence of conformal invariance of the steady state and show the data for MI, $\MIE_X$, and $\MIE_Z$ in Fig.~\ref{fig:measurementonly}(d). 
We see that the bound Eq.~(\ref{eq:MIE_stab_bound}) is again almost saturated and the exponent for MI is the same as the exponent for $\MIE_X$ within the margin of numerical error. This serves as another illustration of the bound Eq.~(\ref{eq:CFTbound}).

\subsection{Measurement-driven criticality with sign structure}

\begin{figure*}
    \centering
    \hspace{-0.08\linewidth}
    \subfloat[]{
    \includegraphics[width=0.40\linewidth]{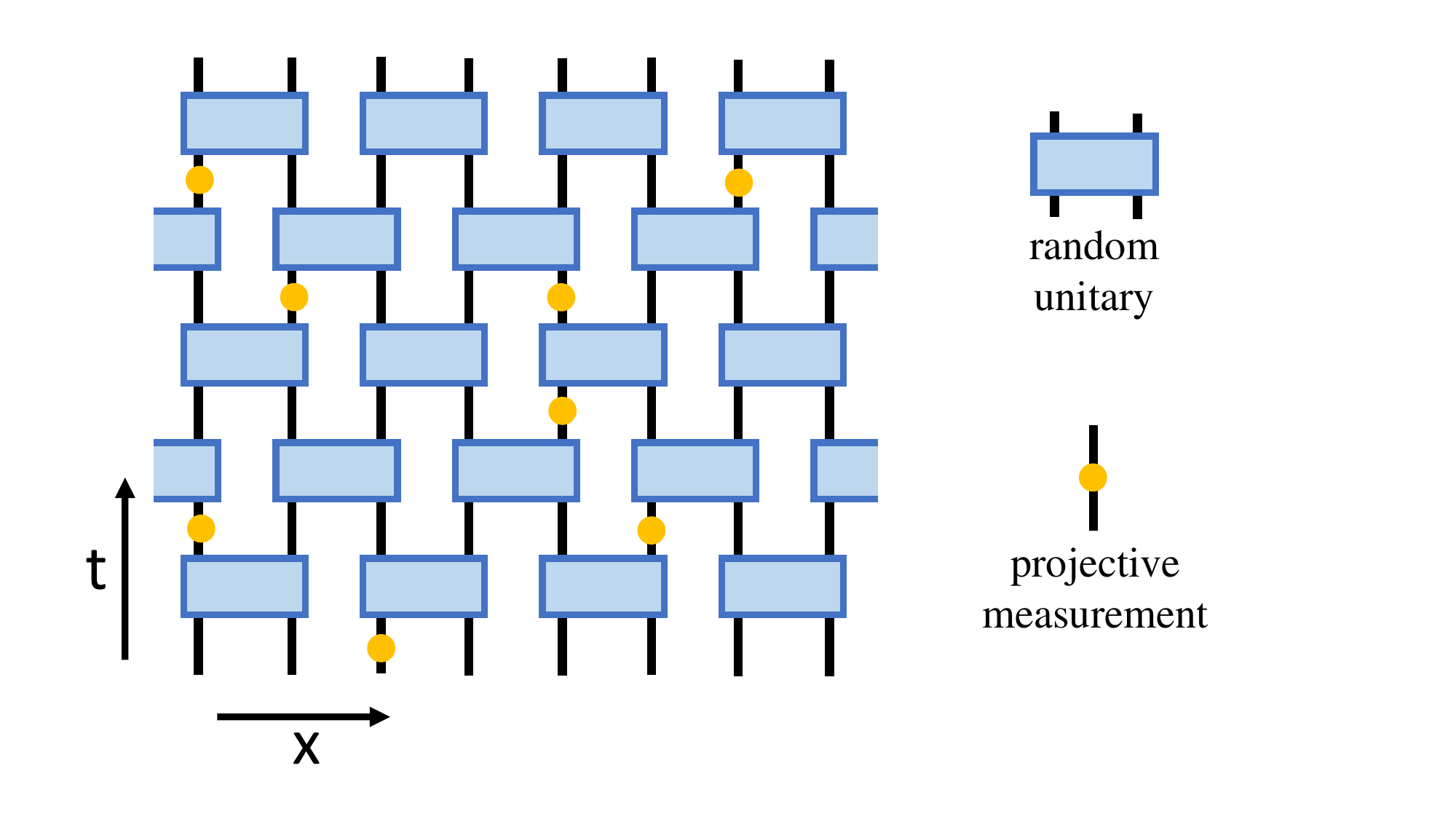}
    }
    \hspace{-0.08\linewidth}
    \subfloat[]{
    \stackinset{l}{0.06\linewidth}{b}{0.03\linewidth}
    {\includegraphics[width=0.07\linewidth]{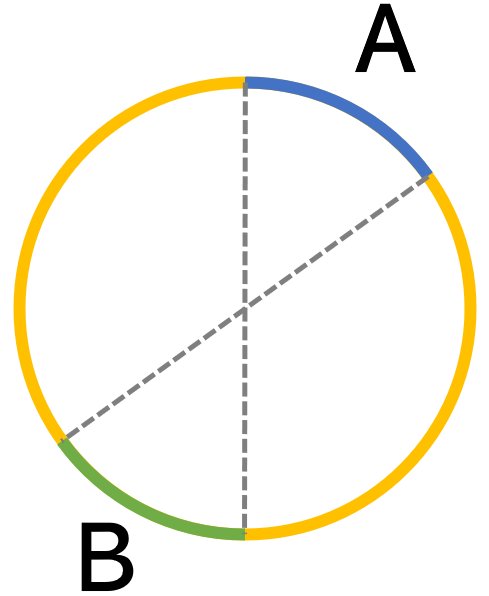}}
    {\includegraphics[width=0.32\linewidth]{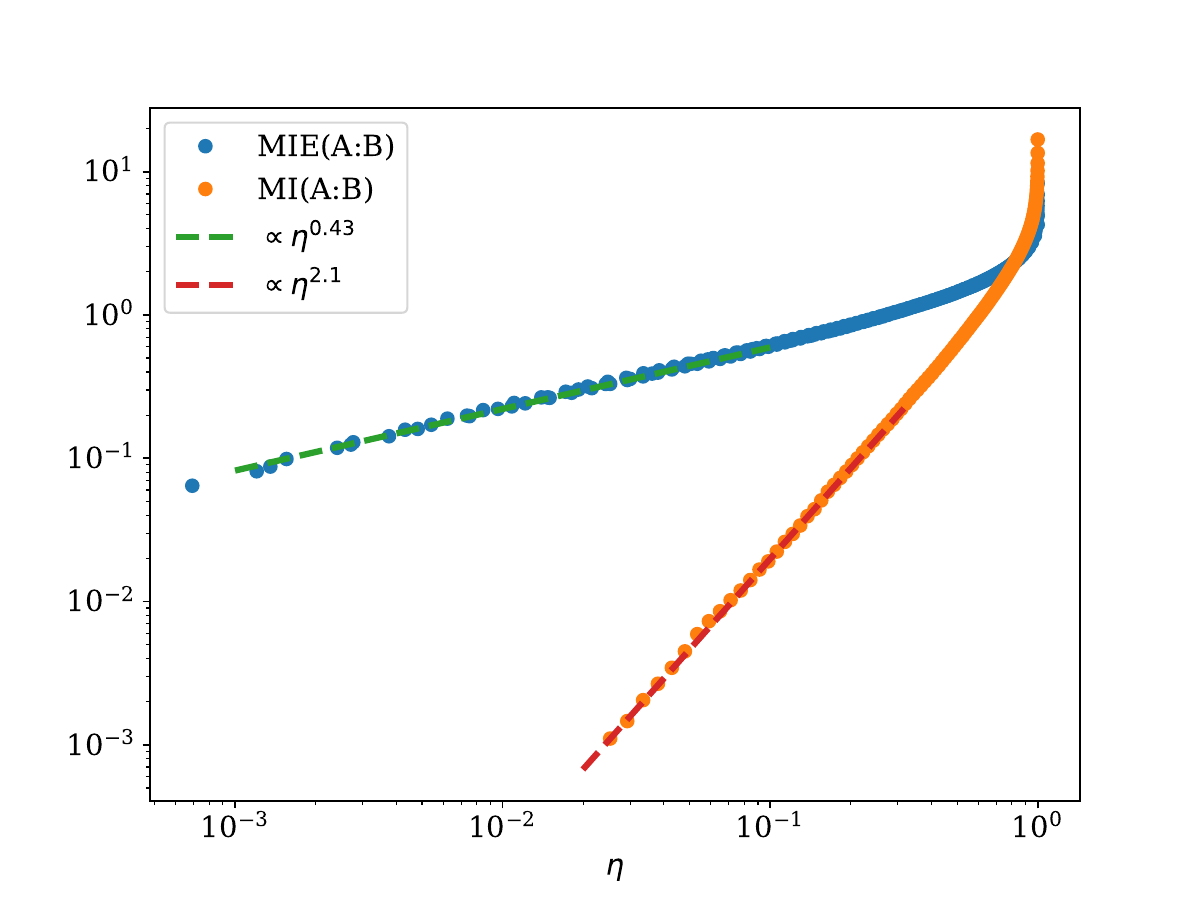}}
    }
    \subfloat[]{
    \stackinset{r}{0.03\linewidth}{b}{0.06\linewidth}
    {\includegraphics[width=0.10\linewidth]{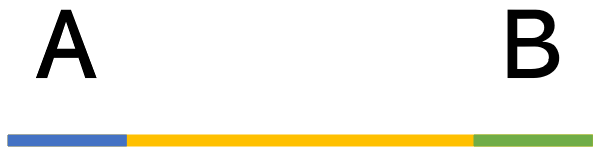}}
    {\includegraphics[width=0.30\linewidth]{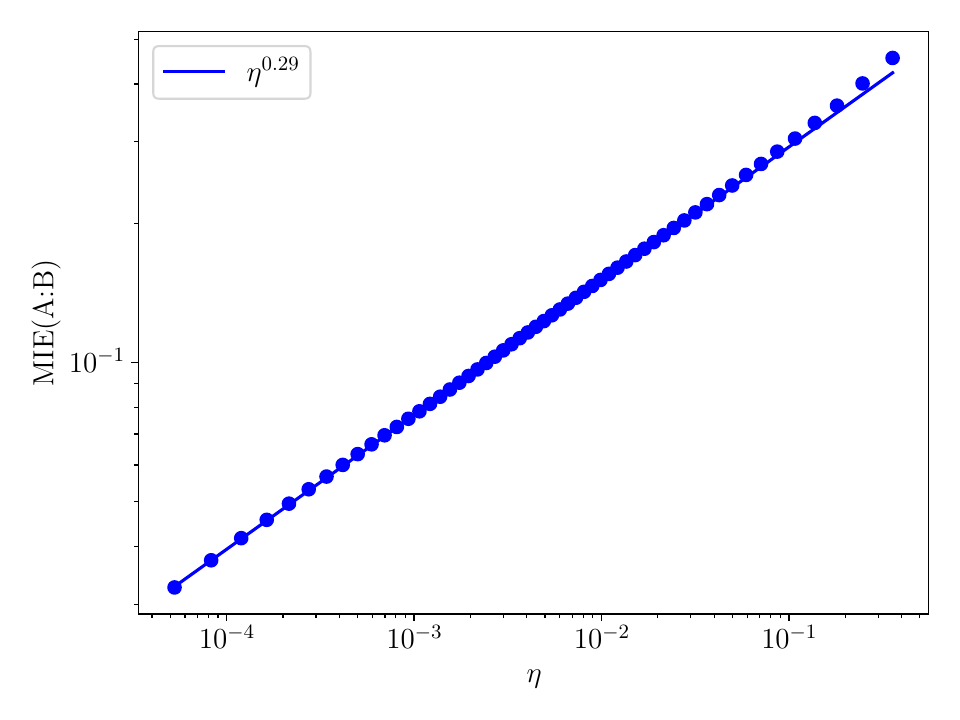}}
    }
    \caption{(a) Hybrid quantum circuits with random unitary gates (blue boxes) in a brick-wall arrangement and single-site measurements (yellow dots) performed with probability $p$. (b) MI and MIE between $A$ and $B$ in Clifford-MIPT, with $L=512$ (c) MIE between $A$ and $B$ in Haar-MIPT, with $L=64$.}
    \label{fig:CliffordHaarMIE}
\end{figure*}

In contrast to the above measurement-only dynamics, hybrid circuits involving random unitaries interspersed with single-site measurements performed with probability $p$ generate steady states with a non-trivial sign structure.  For large $p$ the frequent measurements produce area-law entangled steady states and for small $p$, the unitaries dominate and produce volume-law entangled steady states.  There is a critical value $p_c$ at which a MIPT occurs; at this point, the steady state for one-dimensional systems has logarithmic scaling of entanglement and is described by a CFT \cite{PhysRevB.100.134306, li2021conformal}.

Previous numerical and theoretical studies \cite{li2021statistical, zabalo2022operator} suggest that critical properties of MIPT depend on the choice of the ensemble of random unitary gates. Here we consider two kinds of ensembles: Clifford random unitary gates, and Haar random unitary gates. Their corresponding MIPTs are referred to as Clifford-MIPT and Haar-MIPT, respectively. In both cases we compute the MIE at the critical point averaged over the steady states in the time evolution and measurement outcomes. 
Note that in the MIPT setting, MIE is always calculated after a layer of random unitaries, so the measurement basis is not important.  In our simulations, we fix all measurement basis to $Z$.

For the Clifford-MIPT, we simulated the dynamics over a periodic 1D qubit chain with $L=512$ qubits, using the algorithm proposed in \cite{aaronson2004improved}. We observed that at the critical point $p_c=0.16$, $\text{MIE}\propto \eta^{0.43}$ and $\text{MI}\propto \eta^{2.1}$ when the cross ratio $\eta$ is small, agreeing with numerical results reported in \cite{li2021conformal, li2019measurement}. 
Thus, $h^{\mathrm{Clifford}}_{\text{MI}}>h^{\mathrm{Clifford}}_{\text{MIE}}$, and the violation of the bound Eq.~(\ref{eq:CFTbound}) serves as a direct signature of the non-trivial sign structure in this case.

For the Haar-MIPT, we use matrix-product-state methods with maximal bond dimension $D=1200$ to simulate the time evolution and compute the MIE for system sizes up to $L=64$. 
The truncation error is of the order of $10^{-6}$ at each evolution step, and the the statistical error for MIE is less than $10^{-3}$. (See Appendix~\ref{app:sampling} for details of the sampling algorithm.)

In Fig,~\ref{fig:CliffordHaarMIE}(c), we find that MIE is a power-law function of the cross ratio $\eta$ for small $\eta$, and the power $h^{\mathrm{Haar}}_{\MIE}\approx 0.29$. 
This power does not rely too much on the precision of the critical point and the system size. 
For example, for $L=48$ and $p=0.17$ (not shown in the figure), we find a consistent exponent $h^{\mathrm{Haar}}_{\MIE}\approx 0.30$, while in Fig.~\ref{fig:CliffordHaarMIE}(c), we have $h^{\mathrm{Haar}}_{\MIE}\approx 0.29$ for $L=64$ and $p=0.17$.  In contrast, the mutual information exponent for Haar-MIPT is $h^{\mathrm{Haar}}_{\text{MI}}\approx 2.0$ \cite{nahum2018hybrid, zabalo2020critical}. 
The fact that MIE for Haar-MIPT is even more long-ranged than for Clifford-MIPT ($h^{\mathrm{Haar}}_{\MIE}<h^{\mathrm{Clifford}}_{\MIE}=0.43$) reflects the fact that Haar sign structure is even more complex than the Clifford sign structure, which is discrete $(\pm 1, \pm i)$.

Note that $h^{\mathrm{Clifford}}_{\text{MI}} \approx h^{\mathrm{Haar}}_{\text{MI}}\approx 2.0$, and thus mutual information fails to distinguish the two universality classes.  In contrast, we find that the MIE power laws are able to distinguish the two.

As a brief aside, in Ref.~\cite{li2021conformal}, it was conjectured based on conformal invariance that $h_{\MIE}$ is the same as half of the surface critical exponent $\eta_{||}$, which was numerically shown to hold for Clifford MIPT. 
If the conjecture were true in the Haar random MIPT, it would imply $\eta_{||}\approx 0.58$. 
Nevertheless, this does not agree with a previous exact diagonalization calculation which found $\eta_{||}\approx 0.39$. 
We therefore conclude that either the conjecture does not hold for Haar MIPT or the previous exact diagonalization calculation has significant finite-size corrections.

\section{MIE Bound for Sign-Free States of Qubits}

After demonstrating the effects of being sign-free in the dynamical setting, we now turn our attention to ground states of local Hamiltonians.  First, we present a bound on MIE assuming that the system consists of qubits and $A,B$ are each single qubits.

Following Ref.~\cite{Popp2005}, it is convenient to consider concurrence \cite{PhysRevLett.80.2245} as a measure of entanglement between two qubits.  Given a two-qubit state $\ket{\psi_{AB}}$, the concurrence is defined as $|\langle \psi^*_{AB}|Y_A Y_B|\psi_{AB}\rangle|$, where $Y$ is the Pauli-Y operator.  

Consider a system of qubits partitioned into single qubits $A,B$, and their complement $C$, and let $\ket{\psi}=\sum_{abc} \psi_{abc} |abc\rangle$ be a state that is sign-free ($\psi_{abc} \geq 0$).  Then the measurement-induced concurrence is
\beq
\MIC&=&\sum_c p_c |\langle \psi^*_{AB,c}|Y_A Y_B|\psi_{AB,c}\rangle| \\
&=&\sum_c \left|\sum_{ab,a'b'} \psi_{abc} \psi_{a'b'c} \langle a'b'|Y_A Y_B |ab\rangle\right|
\eeq
Note that $\langle a'b'|Y_A Y_B |ab\rangle=-1$ if $ab=00,a'b'=11$ (or vice versa) and $\langle a'b'|Y_A Y_B |ab\rangle=1$ if $ab=01,a'b'=10$ (or vice versa).

Let us first consider the special case that the system has an Ising symmetry $\prod_j Z_j$.  Then for a fixed outcome $c$, the complement $ab$ has fixed parity (either even or odd).  
Thus for a given $c$, $\langle a'b'|Y_A Y_B |ab\rangle$'s are either all $+1$ or all $-1$, and the absolute value is taken.
We therefore have
\beq
\MIC&=&\sum_c \sum_{ab,a'b'} \psi_{abc} \psi_{a'b'c} \langle a'b'|X_A X_B |ab\rangle \\ \notag
&=& \langle \psi|X_A X_B|\psi\rangle~,
\eeq
which was first noted in Ref.~\cite{Popp2005}.
For a general sign-free wavefunction without assuming an Ising symmetry, the right hand side is an upper bound since the opposite signs can only reduce the absolute value.
Hence, in general
\beq
\MIC\leq \langle \psi |X_A X_B|\psi\rangle~.
\eeq
In addition, using the reality of the wavefunction, we have $\MIC\geq |\langle \psi|Y_A Y_B|\psi\rangle|$.
We therefore see that the MIC is bounded between two (disconnected) correlations in the wavefunction before measurement.  Furthermore, because von Neumann entropy is upper bounded by concurrence, $\mathrm{MIE}\leq \mathrm{MIC}$ and thus
\beq
\MIE \leq \langle \psi|X_A X_B|\psi\rangle. \label{eq:qubitbound}
\eeq
Note that with Ising symmetry, the right hand side is a connected correlation because $\langle X_A \rangle,\langle X_B\rangle$ vanish due to symmetry.

\section{ MIE in critical ground states}

We illustrate the above bound in the one-dimensional critical transverse field Ising model and tri-critical Ising model, both of which have an Ising symmetry and ground states which are sign-free. 
In Appendix~\ref{app:threestatePotts}, we also evaluate MIE in the 1d critical three-state Potts model, which is also sign-free; despite having local Hilbert space dimension $3$ and no Ising symmetry, we find that the decay of MIE is still bounded by the existing correlations before measurement.   
Finally, we compute MIE in a 1d gapless symmetry-protected topological (SPT) model, which has a nontrivial sign structure, and we find that the MIE is unrestricted in that case.

In all of the models considered below, we first use the density matrix renormalization group method to find the ground states of the models, and then sample over different measurement outcomes to calculate the MIE as shown in Appendix~\ref{app:sampling}.

\subsection{1d  transverse-field Ising model}

We first consider the critical model 
\begin{equation}
    H _{\text{Ising}}= -\sum_{j} X_j X_{j+1} - \sum_j Z_j~,
\end{equation}
with $L$ sites and periodic boundary conditions $L+1 \equiv 1$.
This model has the Ising symmetry $P \equiv \prod_j Z_j$.

\begin{figure}
    \centering
    \includegraphics[width=\columnwidth]{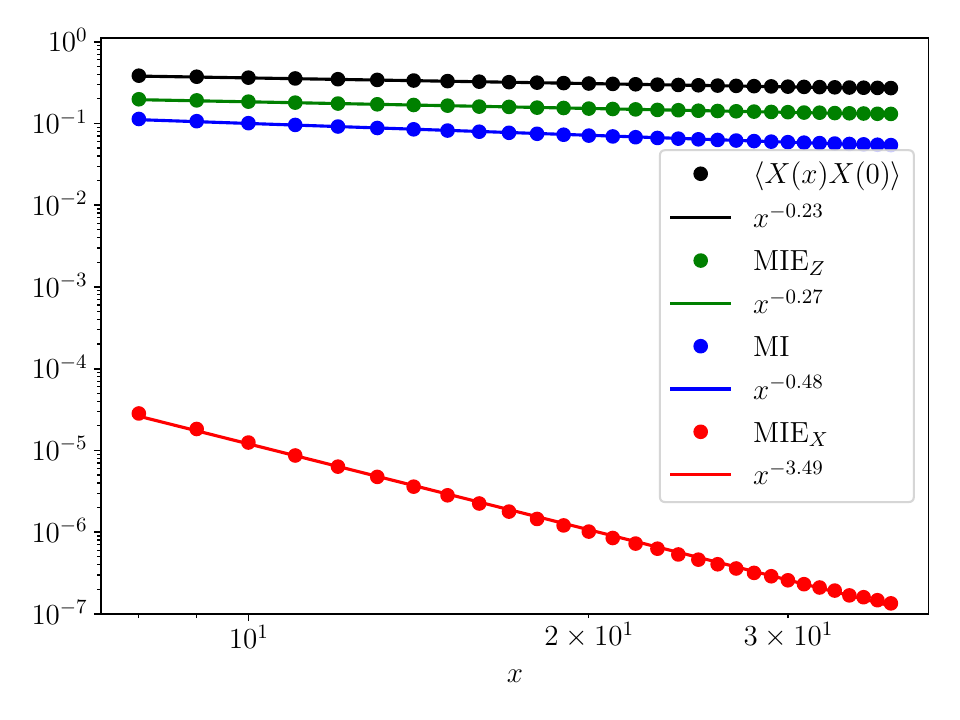}
    \caption{MI, MIE and the relevant correlation function in the critical quantum Ising model $H_{\text{Ising}}$ with system size $L=128$, where $x$ is the separation between qubits $A,B$.}
    \label{fig:MIE_Ising}
\end{figure}

It is well known that the mutual information of the critical Ising ground state scales as $x^{-1/2}$~\cite{calabrese2009} as shown in Fig.~\ref{fig:MIE_Ising}, while the slowest decaying correlation function is $\la X_A X_B \ra \sim x^{-1/4}$, where $x$ is the separation between the two qubits $A,B$.

After measurement, MIE still has a scaling form $\sim x ^{-\alpha}$. 
In particular, we observe that when measured in the $X$-basis, $\alpha \approx 3.5$ which is significantly larger than the slowest exponent; when measured in the $Z$-basis, the exponent of the MIE $\alpha \approx 0.27$, which is very close to the slowest scaling scaling $\la X_A X_B \ra \sim x^{-1/4}$, consistent with the bound Eq.~(\ref{eq:qubitbound}).   

\begin{figure}
    \centering
    \includegraphics[width=\columnwidth]{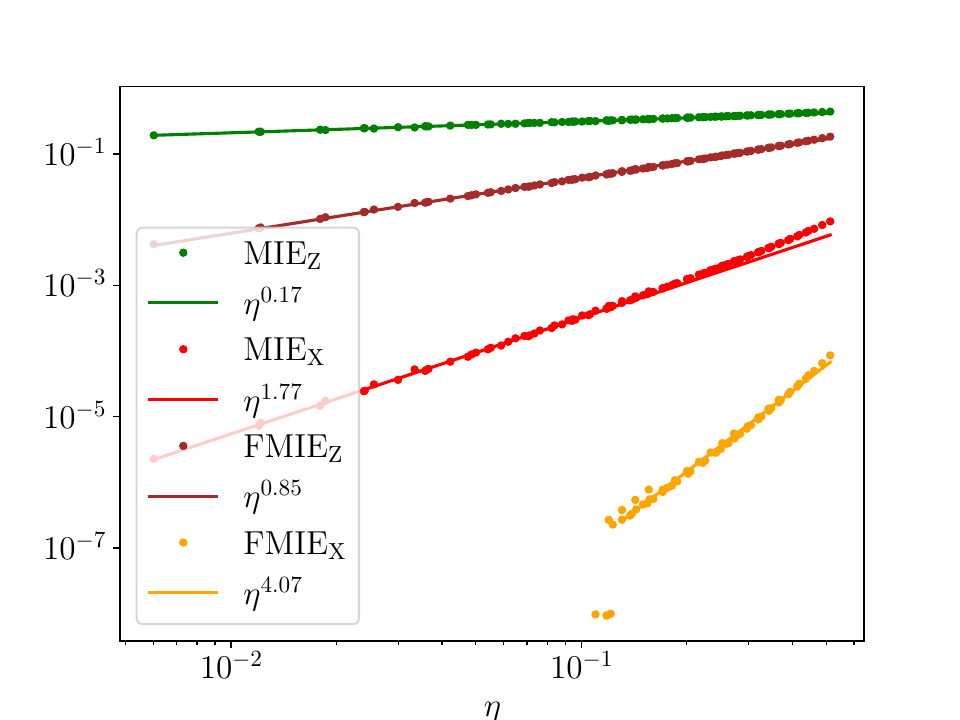}
    \caption{MIE and forced-measurement-outcome MIE (FMIE) versus cross-ratio $\eta$ in the critical quantum Ising model $H_{\text{Ising}}$ with system size $L=128$. $A$ and $B$ are intervals starting from two antipodal points, and their lengths span from 2 to 25.}
    \label{fig:MIE_eta_Ising}
\end{figure}

We also vary the sizes of $A,B$, and we find that remarkably MIE appears to be a function of the cross-ratio $\eta$ only and displays a simple power law $\eta^{\alpha/2}$ when $\eta$ is small (see Fig.~\ref{fig:MIE_eta_Ising}).  Note that the exponents differ slightly from those of Fig.~\ref{fig:MIE_Ising}, likely due to finite size effects and the larger $\eta$'s used.  
In Fig.~\ref{fig:MIE_eta_Ising} we also compare $\MIE$ with ``forced-measurement-outcome MIE'' (FMIE), or ``entanglement entropy after selective outcome'', studied in \cite{Najafi2016,Stephan2014,Rajabpour_2016}. In particular, $\mathrm{FMIE}_Z$ corresponds to measurement in $Z$ basis with $|0 \dots 0\rangle$ outcome (here $Z|0\rangle = |0\rangle$) and $\mathrm{FMIE}_X$ corresponds to measurement in $X$ basis with $|+\dots+\rangle$ outcome. Fig.~\ref{fig:MIE_eta_Ising} suggests that $\MIE$ is not dictated by a few measurement outcomes, which we further support in Appendix \ref{app:distribution}. Moreover, in \cite{Najafi2016,Stephan2014} it is claimed that $\mathrm{FMIE}_Z$ corresponds to the entanglement after imposing the free conformal boundary condition \cite{Cardy1986} in measurement region $C$ while $\mathrm{FMIE}_X$ corresponds to imposing fixed boundary conditions. Fig.~\ref{fig:MIE_eta_Ising} suggests that $\MIE_Z$ and
$\MIE_X$ may be characterized by similar quantities in conformal field theory, though the precise boundary condition is currently not known.

\subsection{1d tri-critical Ising model}
We also consider the O'Brien-Fendley model \cite{PhysRevLett.120.206403},
\begin{equation}
    H_{\text{OF}}=H_{\text{Ising}} +\lambda \sum_{j=1}^{N}(Z_jX_{j+1}X_{j+2}+X_jX_{j+1}Z_{j+2})~
\end{equation}
where we again use periodic boundary conditions and the parameter $\lambda = 0.428$ so that the ground state is at criticality in the tricritical Ising universality class.  The ground state is also sign-free for this parameter.
Note that the slowest decaying correlation function is given by $\la X_A X_B \ra \sim x^{-0.15}$ and the mutual information has a scaling $x^{-0.3}$~\cite{calabrese2009}.

The MIE measured in the $X$-basis has the scaling $x^{-\alpha}$ where $\alpha \approx 3.5$ while when measured in $Z$-basis, $\alpha \approx 0.14$, which is within numerical error of the $\la X_A X_B \ra$ exponent $0.15$. Again, this is consistent with the bound Eq.~(\ref{eq:qubitbound}).

\begin{figure}
    \centering
    \includegraphics[width=\columnwidth]{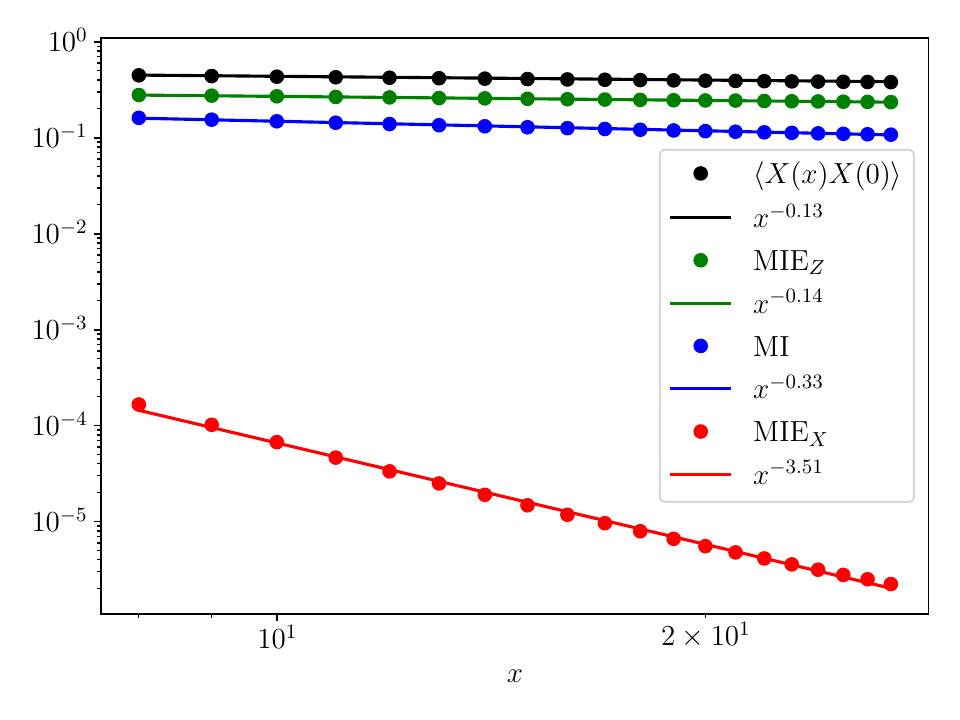}
    \caption{MI, MIE and the relevant correlation function in the tri-critical Ising model $H_{OF}$ with system size $L=96$, where $x$ is the separation between qubits $A,B$.}
    \label{fig:MIE_TCI}
\end{figure}

\subsection{1d gapless SPT model}

\begin{figure}
    \centering
    \includegraphics[width=\columnwidth]{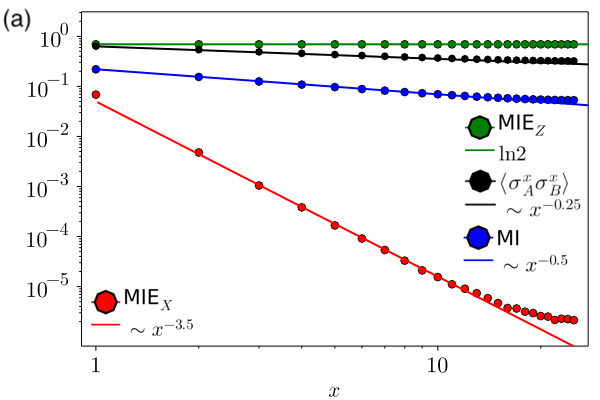}
    \includegraphics[width=\columnwidth]{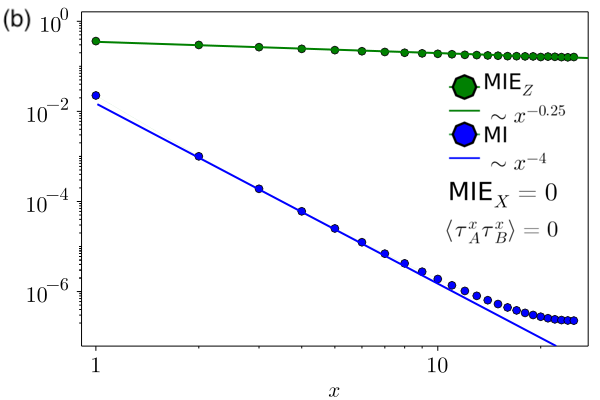}
    \caption{(a) MI, MIE and the relevant correlations in the gapless SPT model $H_{gSPT}$ when $A,B$ are $\sigma$ spins with the system size $L=100$, where $x$ is the separation between $A,B$. (b) Same as (a) but when $A,B$ are $\tau$ spins.  }
    \label{fig:MIE_gSPT}
\end{figure}

In this subsection, we show that it is possible to have a critical point with unrestricted MIE, with no such bound as Eq.~(\ref{eq:qubitbound}), when the wavefunction is not sign-free.
We consider the 1d model proposed in Ref.~\cite{scaffidi}, consisting of $L$ spins labeled as $\sigma$ on the integer sites $j =1 \dots L$ and $L$ spins labeled as $\tau$ on the sites $j+\frac{1}{2}$ where $j =1 \dots L$.
The Hamiltonian is 
\beq
    H_{gSPT} =&
    -&\sum_j(\tau^x_{j-\frac{1}{2}}\sigma^z_j\tau^x_{j+\frac{1}{2}}  
    +\sigma^x_{j-1}\tau^z_{j-\frac{1}{2}}\sigma^x_{j}) \notag \\ 
    &-&\sum_j \sigma^x_j\sigma^x_{j-1}
\eeq
with the periodic boundary condition $L+j \equiv j$.

The two pieces of the Hamiltonian are the cluster state and ferromagnetic interactions, which compete to yield a critical point.  Despite being gapless, with open boundary conditions there are boundary modes protected by the global $Z_2 \times Z_2$ symmetry from flipping each sublattice; hence the name gapless SPT (see Ref.~\cite{scaffidi} for more details).  For our purposes, it is important that the ground state wavefunction is not sign-free in either the $Z$ or the $X$ basis, just like its gapped counterpart-- the cluster state.

In Fig.~\ref{fig:MIE_gSPT}, we numerically calculate MIE, MI, and relevant correlation functions from the ground state wavefunction obtained from density matrix renormalization group method with $L=100$.  As before, we take $A,B$ to be single qubits, either both on the $\sigma$ sublattice (Fig.~\ref{fig:MIE_gSPT}(a)) or $\tau$ sublattice (Fig.~\ref{fig:MIE_gSPT}(b)). 
All of the power laws observed numerically are derived in the Appendix~\ref{app:gSPTresults}.

A striking feature is that, if $A$ and $B$ are on the $\sigma$ sublattice, then the MIE after measuring in the $Z$-basis does not decay at all and is equal to $\ln 2$.  This is in stark contrast to MI and correlation functions which decay with power laws.  This constant MIE can be understood in a similar way as the gapped cluster state discussed earlier: \{$\sigma^x_j \tau^z_{j+\frac{1}{2}} \sigma^x_{j+1}$\} for all $j$ are conserved quantities and can be multiplied to create a string order parameter as in the cluster state. (Note the different Ising-symmetry convention.)
After the $Z$ measurements are performed, one again is left with an EPR pair.  
On the other hand, if $A$ and $B$ are on the $\tau$ sublattice, then $\MIE_Z$ decays with power $0.25$, slower than any correlation function on the $\tau$ sublattice. 
Both cases illustrate that, without the sign-free constraint, MIE can be more long-ranged than any correlation function in the wavefunction before measurement.

\section{Summary and Discussion}

The main claim we proposed in this work is that for sign-free states, the decay of MIE in the separation of $A$ and $B$ cannot be slower than that of correlations in the wavefunction before measurement. 
We considered two types of systems: stabilizer states and spin chains, with a focus on quantum criticality in both.

For the stabilizer states, we first noted that a sign-free stabilizer state is equivalent to a CSS code -- the generators of the stabilizers can be chosen to be consist of pure-X and pure-Z stabilizers. 
We then proved that $\MIE\leq \MI$ for a sign-free stabilizer state, which supports our proposed claim.
A technical way of understanding this result is that measurements in the $Z$-basis cannot generate $X$-type stabilizers which are required to have entanglement between $A$ and $B$ after measurements. 
Therefore MIE can be nonzero only if the $X$-type generator already exists between $A$ and $B$, which requires existing correlation or MI before the measurements.

We then applied this bound to sign-free measurement-induced critical points and showed that it constrains the scaling dimensions of certain boundary condition changing operators ($h_{\MIE}\geq h_{\MI})$.  We demonstrated these bounds in the steady states arising from 
$X$-$ZZ$ and $XX$-$ZIZ$ measurements.  In contrast, we showed that in critical points emerging from the Clifford-random/Haar random gates plus measurements, which are not sign-free, the MIE can decay slower than the MI.  Furthermore, the different values of $h_{\MIE}$ for Clifford and Haar random hybrid circuits serve as a way to distinguish the two universality classes. 

Next, we showed that for sign-free wavefunctions of qubits, the MIE between two qubits $A,B$ is upper bounded by a two-point correlation $\la X_A X_B \ra$, and we illustrated these results for the transverse-field Ising and tri-critical Ising model. 
In contrast, critical states with sign structure can have unrestricted MIE, as we demonstrated for the 1d gapless SPT.

There are many open questions worth exploring.  For a sign-free wavefunction, if one only assumes that it is the ground state of a local Hamiltonian, is there a general bound between MIE and MI/correlation functions?  Our results for the spin chain ground states (including the three-state Potts model in Appendix~\ref{app:threestatePotts}) suggest that the MIE decays with the same power as the slowest two-point correlation function and is a function of the cross-ratio only; how does one understand these results, or more generally interpret MIE, in the conformal field theories describing these critical points?

It would also be interesting to relate our work to intrinsic sign structure, in both topological order and critical states.  For example, it is natural to conjecture that the critical steady states of the hybrid circuits (with unitaries and measurements) have intrinsic sign structure that no finite depth unitary can remove.  Similarly, it is plausible that the 1d gapless SPT has a symmetry-protected sign problem like the gapped cluster state it descends from.  Are there gapless ground states in 1d with an intrinsic sign problem?  If so, how is the intrinsic sign structure manifested in the CFT data if there is conformal symmetry? And how can MIE be used to diagnose such intrinsic sign problems?  We leave these questions for the future. 

\begin{acknowledgments}
We thank David Gosset and Yaodong Li for valuable discussions and feedback.
This work was supported by Perimeter Institute, NSERC, and Compute Canada.
Research at Perimeter Institute is supported in part by the Government of Canada through the Department of Innovation, Science and Economic Development Canada and by the Province of Ontario through the Ministry of Colleges and Universities. Y.~Z. is supported by the Q-FARM fellowship at Stanford University.
\end{acknowledgments}

\bibliography{ref.bib}

\newpage
\appendix
\onecolumngrid
\section{Structure of the sign-free stabilizer state}\label{app:sign-free stablizer}
In this appendix, we show that the structure of the stabilizer generators for a sign-free stabilizer state is equivalent to that for a CSS code, namely the generators of the stabilizer group can be chosen as pure-$X$ and pure-$Z$ Pauli strings. 
Assume we have $n$ qubits and the stabilizers are generated by the generators $S = \la g^1 \dots g^n \ra$, stabilizing the state $|\psi \ra$. 
The generators can be represented using a binary matrix. 
Here we use the convention $g^{(m)} = \alpha^{(m)}\prod_{i=1}^{n} X_i^{v_i^{(m)}} \prod_{i=1}^{n} Z_i^{w_i^{(m)}}$, where $v_i^{(m)}$ and $w_i^{(m)} = 0, 1$ and $\alpha^{(m)} = \pm 1, \pm i$. 
The check matrix is therefore 
\begin{equation}
    G = \left(\begin{array}{c|c} 
	V & W 
    \end{array}\right)~, 
\end{equation}
where $V$ and $W$ are $n \times n$ matrices with the matrix elements $V_{mi}=v_i^{(m)}$ and $W_{mi}=w_i^{(m)}$.
One can form a ``canonical'' basis of the generators by performing Gaussian eliminations on the check matrix $G$ (and update the signs $\alpha^{(m)}$ accordingly), resulting in a row-echelon form $G'$.
Assume $G'$ has the form 
\begin{equation}
    G' = \left(\begin{array}{c|c} 
	A & B \\
	O & C
    \end{array}\right)~, 
\end{equation}
where $A$ and $B$ are $r \times n$ binary matrices, $O$ is an $(n-r) \times n$ zero matrix and $C$ is an $(n-r) \times n$ matrix. 
And assume that the updated signs are $(\beta^{(1)}, \beta^{(2)} \dots \beta^{(r)},\beta_Z^{(1)}, \dots \beta_Z^{(n-r)})$.
That is to say, there are $k=n-r$ generators that are purely Pauli-Z denoted by $g_z^{(1)} \dots g_z^{(k=n-r)}$, and they generate the stabilizer set $S^Z = \la g_z^{(1)} \dots g_z^{(k)} \ra$. 

Let us show that $|\psi\ra$ has $2^r$ nonzero coefficients in the $Z$-basis $|\bfsigma\ra$.  
Consider the vector space $V_S^{Z}$ which is stabilized by $S^Z$. 
We have $\dim(V_S^{Z}) = 2^{r}$, and it is natural to use a subset of $|\bfsigma\ra$ as the basis for $V_S^{Z}$.
Since $S^{Z}$ is a subgroup of the stabilizer group $S$, we have $|\psi \ra \in V_S^{Z}$. 
More explicitly, 
\begin{equation}
|\psi\ra = \sum_{\bfsigma \in V_S^{Z}} \psi(\bfsigma)|\bfsigma\ra ~.
\end{equation}
On the other hand, $\psi(\bfsigma)$ has to be nonzero, which we show below. 
Define $\bar{S}^Z = \la g_1^{'} \dots g_r^{'} \ra$ which is the stabilizer set generated by the non pure-Z generators. 
Consider a stabilizer $g \in \bar{S}^Z$, and we decompose $g=\beta g_x * g_z$, where $\beta$ is the sign and $g_x$ ($g_z$) are Pauli-X (Pauli-Z) part of the generator, and we allow $g_z = I$.
Consider $g |\psi \ra $, or 
\begin{equation}
    g \sum_{\bfsigma \in V_S^Z}\psi(\bfsigma)|\bfsigma\ra = \sum_{\bfsigma \in V_S^Z} \alpha_g(\bfsigma)\psi(\bfsigma)|g(\bfsigma)\ra~,
\end{equation}
where $|g(\bfsigma)\ra = g_x |\bfsigma\ra$ and the additional sign $\alpha_g(\bfsigma) = \beta \cdot \gamma(\bfsigma)$ and $\gamma(\bfsigma) = \pm 1$ is from $g_z|\bfsigma\ra = \gamma(\bfsigma)|\bfsigma\ra$.
Since $g |\psi \ra = |\psi \ra$, we have $\alpha_g(\bfsigma)\psi(\bfsigma)=\psi(g(\bfsigma))$.
Now suppose $\exists |\bfsigma^*\ra \in V_S^Z$ such that $\psi(\bfsigma^*)=0$, then we also get $\psi(g(\bfsigma^*))=0$.
Looping over all the $g \in \bar{S}_Z$ will in fact give us $|\psi \ra =0$ which is absurd. 
Therefore, all the coefficients has to be nonzero, i.e., $\psi(\bfsigma) \neq 0$.  

Now it is easy to see the condition for a sign-free stabilizer state.
We would require $\alpha_g(\bfsigma) = \beta * \gamma(\bfsigma) = 1$ for all $g = \beta g_x g_z \in \bar{S}^Z$ and $\bfsigma \in V_S^Z$.
Assume both $g_z, -g_z \not\in S^Z$, then 
\begin{equation}
    g|\psi\ra = |\psi\ra= \frac{1}{\sqrt{2^r}} \beta g_x (\sum_{\bfsigma_+} |\bfsigma_+ \ra - \sum_{\bfsigma_-} |\bfsigma_- \ra), 
\end{equation}
where $g_z | \bfsigma_{\pm} \ra= \pm | \bfsigma_{\pm} \ra$, which would not be sign free.
Therefore, we have to have $g_z \in V_S^{Z}$ and $\beta=1$ or $-g_z \in V_S^{Z}$ and $\beta=-1$. 
This also implies that, a sign free stabilizer state can be expressed with pure-X stabilizers with all positive signs and pure-Z stabilizers (with possibly $\pm$ signs), since the $\beta g_z$ part of $g \in \bar{S}^Z$ is in $S^Z$.

Finally, we quickly comment on the relation between a sign-free stabilizer state and some ``simply signed'' states. 
Consider a type of stabilizer state which is stabilized by pure-X and pure-Z stabilizer, but possibly with different signs. 
It can be transformed to a sign-free stabilizer state by operating single qubit gates $U=\prod_j Z_j$ where $j$ is on the ``pivot-X'' qubits if it is the pure-$X$ generator with a minus sign.
That is to say, after the Gaussian elimination and back-substitution of the check matrix, if some pure-$X$ stabilizer generator has a minus sign, then we can operate $U_j=Z_j$ on the qubit where the ``pivot'' is located. 
Since the operations are only one qubit gates, the entanglement structure is preserved.

\section{Structure theorem for sign-free stabilizer states}\label{app:structure_thm}

In this appendix, we show that for any sign-free stabilizer state (CSS code) $|\psi\rangle$, there exist local unitaries $U_A$, $U_B$, $U_C$ that preserve the sign-free structure (that map strings of Pauli $X$ ($Z$) operators to strings of Pauli $X$ ($Z$) operators) such that  
\beq
    \label{eq:structural_thm_CSS)_app}
  U_A U_B U_C \ket{\psi}
    && =
    \ket{\mathrm{GHZ}}_{ABC}^{\otimes g_{ABC}} \otimes \ket{\mathrm{GHZ'}}_{ABC}^{\otimes g'_{ABC}} \\
    &&\otimes
    \ket{\mathrm{EPR}}_{AB}^{\otimes e_{AB}} \otimes
    \ket{\mathrm{EPR}}_{BC}^{\otimes e_{BC}} \otimes
    \ket{\mathrm{EPR}}_{CA}^{\otimes e_{CA}} \nonumber
    \\ 
    &&\otimes \ket{0}_A^{\otimes s_{A}}
    \otimes \ket{0}_B^{\otimes s_{B}}
    \otimes \ket{0}_C^{\otimes s_{C}} \otimes \ket{+}_A^{\otimes s'_{A}}
    \otimes \ket{+}_B^{\otimes s'_{B}}
    \otimes \ket{+}_C^{\otimes s'_{C}}, \nonumber
\eeq
where $|\mathrm{GHZ}\rangle = \frac{1}{\sqrt{2}}(|000\rangle + |111\rangle)$, $|\mathrm{EPR}\rangle = \frac{1}{\sqrt{2}}(|00\rangle + |11\rangle)$ and $|\mathrm{GHZ}'\rangle = \frac{1}{\sqrt{2}}(|+++\rangle + |---\rangle)$ and $\ket{\pm}$ eigenstates of $X$: $X\ket{\pm}=\pm \ket{\pm}$. We follow the proof of structure theorem for stablizer states in Ref.~\cite{Bravyi2006}.

\subsection{Preliminaries}

Recall that for a given quibit, the Pauli operators $X,Y,Z$ and Identity operator $I$ can be thought of as elements in a two-dimensional binary linear space $(\mathbb{F}_2)^2$, up to $\pm 1,\pm i$. Moreover, given $n$ qubits, strings of Pauli operators can be thought of as elements in a $2n$-dimensional binary linear space $G_n \equiv (\mathbb{F}_2)^{2n}$. The space $G_n$ has a symplectic structure from the commutation relation of operators, i.e. for two strings of Pauli operators $f, g\in G_n$ we have
\beq
f\cdot g = (-1)^{\omega(f, g)}g \cdot f,
\eeq
where $\omega: G_n \times G_n \rightarrow \mathbb{Z}_2$ is a symplectic form on $G_n$. In the cannonical basis consisting of single-qubit $X$ operators and single-qubit $Z$ operators, $\omega$ can be written as a block-off-diagonal matrix $\left(\begin{array}{cc}
     & I_n \\
     I_n &  \\
\end{array} \right)$.

For any subspace $S \subseteq G_n$, define the dual subspace $S^\perp$ as
\beq
S^\perp \equiv \{f \in G_n:~~~ \omega(f, g) = 0~~~\text{for all}~~g \in S\}.
\eeq
A subspace $S$ is called isotropic if $S \subseteq S^\perp$, i.e., $\omega(f, g) = 0$ for any $f, g \in S$. A subspace $S$ is called self-dual if $S^\perp = S$. For any isotropic (self-dual) subspace $S \subseteq G_n$ one has $\text{dim}(S) \leq n$ ($\text{dim}(S) = n$). 
The stabilizer group of a stabilizer state $|\psi\rangle$ is nothing but a self-dual subspace $S$.

A unitary operator $U$ belongs to the Clifford group, $U \in Cl(n)$, if it maps Pauli string operators to Pauli string operators (up to a sign) under the conjugation. In other words, $U \in Cl(n)$ if there exists a map $u : G_n \rightarrow G_n$ and a function $s: G_n \rightarrow \{+1, -1\}$, such that
\beq
U \cdot f \cdot U^\dagger = s(f) \cdot u(f) 
\eeq
for any $f \in G_n$.  Unitarity of $U$ implies that $u$ is a linear
invertible map preserving the symplectic form $\omega$, i.e.,
\beq
\omega(f, g) = \omega(u(f), u(g))
\eeq
for all $f, g \in G_n$. Such linear maps constitute a binary
symplectic group $Sp(2n, \mathbb{F}_2)$. Moreover, all $u \in Sp(2n, \mathbb{F}_2)$ can be realized through an appropriate choice of $U \in Cl(n)$.

In the context of sign-free stablizer states (CSS codes), we can similarly define the CSS-Clifford group 
\begin{definition}
A unitary operator $U$ belongs to the CSS-Clifford group, $U \in CSSCl(n)$, if it maps pure-$X$ ($Z$) Pauli string operators to pure-$X$ ($Z$) Pauli string operators (up to a sign) under the conjugation. We will sometimes refer to $X$/$Z$ the type of the operator. 
\end{definition}
Choose the canonical basis for $G_n$ in terms of single-qubit $X$ operators and single-qubit $Z$ operators, we see that the corresponding $u: G_n \rightarrow G_n$ of some $U\in CSSCl(n)$ should live in the $SL(n, \mathbb{F}_2)$ subgroup of $SP(2n, \mathbb{F}_2)$ consisting of elements in the form $\left(\begin{array}{cc}
     v &  \\
      & \left(v^{-1}\right)^T \\
\end{array} \right)$, where $v$ is an $n\times n$ binary matrix with unit (nonzero) determinant.

For a sign-free stabilizer state of $n$ qubits, we can choose a CSS basis for the stabilizer group $S$ to consist of pure-$X$ generators, which form a group $S^X$, as well as pure-$Z$ generators, which form a group $S^Z$. Two sign-free stabilizer states are connected to each other by CSS-Clifford unitaries iff their CSS basis have the same number of $Z$-type generators, denoted by $n_Z$, or the same number of $X$-type generators, denoted by $n_X = n-n_Z$. Recall that any stabilizer state can be represented by $|\psi\rangle = U |0\rangle^{\otimes n}$ for some operator $U \in Cl(n)$. Similarly, any sign-free stabilizer state can be represented by $|\psi\rangle = U |0\rangle^{\otimes n_Z}\otimes |+\rangle^{\otimes n_X}$ for some operator $U \in CSSCl(n)$.

\begin{definition}
$M$-partite sign-free stabilizer states $|\psi\rangle, |\psi'\rangle$ are called CSS-LCU-equivalent (where LCU stands for local Clifford unitary) if there exist CSS-Clifford unitaries $\{U_\alpha\in CSSCl(n_\alpha), \alpha\in M\}$, where $n_\alpha$ is the number of qubits for party $\alpha$, such that
$|\psi'\rangle = \left(\otimes_{\alpha\in M} U_\alpha\right) |\psi\rangle$.
\end{definition}

We wish to show that tripartite sign-free stabilizer states $|\psi\rangle$ are CSS-LCU-equivalent to the direct product of $|\mathrm{GHZ}\rangle$, $|\mathrm{GHZ}'\rangle$, $|\mathrm{EPR}\rangle$, $|0\rangle$ and $|+\rangle$. A related concept is the extraction of some (usually simple) state $|\psi'\rangle$ from $|\psi\rangle$.

\begin{definition}
Let $|\psi\rangle$ and $|\psi'\rangle$ be two $M$-partite (sign-free) stabilizer states of $n$ and $m$ qubits, respectively, such that in individual party $\alpha\in M$ the number of qubits $n_\alpha, m_\alpha$ for $|\psi\rangle, |\psi'\rangle$ satisfy $n_\alpha\geq m_\alpha$. The state $|\psi'\rangle$ is extractable from $|\psi\rangle$ if $|\psi\rangle$ is (CSS-)LCU-equivalent to $|\psi'\ra \otimes| \psi''\rangle$ with another $M$-partite stabilizer state $|\psi''\rangle$.
\end{definition}

We will present the proof by first extracting single-qubit states $|0\rangle, |+\rangle$, and then GHZ states $|\mathrm{GHZ}\rangle, |\mathrm{GHZ}'\rangle$, and showing that the rest is CSS-LCU-equivalent to EPR states $|\mathrm{EPR}\rangle$. But first we need a criterion for such extraction.

\subsection{Criterion for extraction}

Following Ref. \cite{Bravyi2006}, first we present the following fact from linear algebra.
\begin{lemma}\label{lemma: basic linear algebra}
Let $f_1,\dots, f_p$ and $f_1',\dots,f_p'$ be two families of vectors (operators) in $G_n$ satisfying the following conditions:
\beq
\omega(f_i , f_j) = \omega(f_i', f_j')~~~~\text{for any}~~1\leq i,j\leq p,
\eeq
\beq
\sum_i x_i f_i = 0 \quad\text{iff}\quad \sum_i x_i f_i' = 0~~~~\text{where $x_i=0$ or $1$.}
\eeq
and $f_i$ and $f_i'$ are of the same ($X$/$Z$)-type for any $1\leq i \leq p$. Then there exists $v \in SL(n, \mathbb{F}_2)$, an $n\times n$ binary matrix, such that $f_i' = v(f_i)$ for all pairs of $X$-type vectors $(f_i ,f_i')$ and $f_i' = (v^{-1})^T(f_i)$ for all pairs of $Z$-type vectors $(f_i, f_i')$.
\end{lemma}

\begin{proof}
Let us call a basis $e_1, \bar{e}_1, \dots , e_n, \bar{e}_n$ of the space $G_n$ canonical if $e_i, i=1,\dots,n$ are of $X$-type, $\bar{e}_i,i=1,\dots,n$ are of $Z$-type and the following condition holds
\beq
\omega(e_i, \bar{e}_j) = \delta_{ij}.
\eeq

We can extend the family $f_1,\dots, f_p$ to a canonical basis $\{e_i, \bar{e}_i\}$ using the Gram-Schmidt orthogonalization algorithm as follows. Demand $e_1$ or $\bar{e}_1$ to be $f_1$ depending on the type of $f_1$. Now suppose that the algorithm has processed $f_1,\dots,f_{i-1}$, and $e_{j_1},\dots,e_{j_s}$, $\bar{e}_{k_1},\dots,\bar{e}_{k_t}$ have been chosen. Here, $j_1,\dots,j_s$ and $k_1,\dots,k_t$ are not necessarily $1,\dots,s$ and $1,\dots,t$ for reasons that will be clear latter. We wish to find an algorithm to process $f_i$ and extend the canonical basis.

Without loss of generality, suppose $f_i$ is of $X$-type. If $f_i$ is a linear combination of $f_1,\dots,f_{i-1}$, then through induction we immediately see that $f_i$ can be written as linear combination of $e_{j_1},\dots,e_{j_s}$.
If not, consider all $l\in \{ k_1,\dots,k_t \}$ such that $\omega(f_i, \bar{e}_l)=1$ and call this set $O_i$. If for all $l\in O_i$, $e_l$ have been chosen, we demand that $e_{j_{s+1}} = f_i - \sum_{l\in O_i} e_l$,
where $j_{s+1}$ can be chosen to be the smallest number not in $j_1,\dots,j_s,k_1,\dots,k_t$, and we have
$
f_i=  \sum_{l\in O_i} e_l + e_{j_{s+1}}.
$
If for one $l_1\in O_i$, $e_{l_1}$ has not been chosen, then we demand that $e_{l_1} = f_i - \sum_{l\in O_i/l_1} e_l$. If for two or more of $l_1,\dots,l_r\in O_i$, $e_{l_1},\dots,e_{l_r}$ have not been chosen, we first make a basis transformation $\bar{e}_{l_2}, \dots, \bar{e}_{l_r} \rightarrow \bar{e}_{l_2} + \bar{e}_{l_1}, \dots, \bar{e}_{l_r} + \bar{e}_{l_1}$ and then demand that $e_{l_1} = f_i -  \sum_{l\in O_i/l_1,\dots,l_r} e_l$. In short, one can get a set of cannonical basis $\{e_i, \bar{e}_i\}$ and write 
\beq
f_j = \sum_{k=1}^n F_{jk} e_k\quad\quad \text{or} \quad\quad f_j = \sum_{k=1}^n \bar{F}_{jk} \bar{e}_k~,
\eeq
where $F$ and $\bar{F}$ are some binary $p \times n$ matrices. 

Clearly, if we apply the same algorithm in parallel to $f_1',\dots,f_p'$, we shall end up with another canonical basis $\{e_i', \bar{e}_i'\}$ such that $f_j' = \sum_{k=1}^n F_{jk} e_k'$ or $f_j' = \sum_{k=1}^n \bar{F}_{jk} \bar{e}_k'$ with the same $F$ and $\bar{F}$. The symplectic group $SL(n, \mathbb{F}_2)$ acts transitively on the set of canonical bases. Thus, 
\beq
e_i' = v(e_i) \quad\quad \text{and} \quad\quad \bar{e_i'} = (v^{-1})^T(e_i)
\eeq
for some $n\times n$ invertible binary matrix $v\in SL(n, \mathbb{F}_2)$. This implies that $f_j' = v(f_j)$ for all $X$-type vectors $f_j, f_j'$ and $f_j' = (v^{-1})^T(f_j)$ for all $Z$-type vectors $f_j, f_j'$.
\end{proof}

In order to state the criterion for extraction, we need an additional definition. 
For any vector $f \in G_n$ and party $\alpha$, denote the projection of $f$ onto the party $\alpha$ by $f_\alpha$.
\begin{definition}
Suppose $n$ qubits are distributed among $M$ parties. Let $S \subseteq G_n$ be a linear subspace. For
each $\alpha \in M$ define the local subspace $S_\alpha \subseteq S$ and the co-local subspace $S_{\hat \alpha} \subseteq S$ as
\beq
S_\alpha = \{f \in S :~~~ f_\beta = 0~~~\text{for all}~~\beta \in M/ \alpha\},
\eeq
and
\beq
S_{\hat\alpha} = \{f \in S :~~~ f_\alpha = 0\}.
\eeq
An $M$-party stabilizer state $|\psi\rangle$ with a stabilizer group $S$ is said to have full local ranks iff all local subgroups of $S$ are trivial: $S_\alpha = \emptyset$ for all $\alpha \in M$.
\end{definition}
In other words, $f \in S_{\hat \alpha}$ iff $f$ acts as the identity on party $\alpha$; $f \in S_\alpha$ iff $f$ acts as the identity on all parties $\beta\in M$ except $\alpha$. 

Following Ref.~\cite{Bravyi2006}, we can state the criterion for extraction.

\begin{lemma}\label{lemma: main lemma}
Let $|\psi\rangle$ and $|\psi'\rangle$ be $M$-party sign-free stabilizer states with stabilizer groups $S$ and $S'$, respectively. The state $|\psi'\rangle$ is extractable from $|\psi\rangle$ iff there exists a linear injective map $T:S' \rightarrow S$ such that
\begin{enumerate}
\item $\omega(T(f)_\alpha, T(g)_\alpha) = \omega(f_\alpha, g_\alpha)$ for all $f, g \in S'$ and $\alpha \in M$;
\item $(T \cdot S')_{\hat{\alpha}} = T \cdot (S'_{\hat{\alpha}})$ for all $\alpha \in M$.
\item $T(f)$ and $f$ always have the same ($X/Z$) type.
\end{enumerate}
\end{lemma}

\subsection{Proof of structure theorem}
Let $|\psi\rangle$ be an $M$-party sign-free stabilizer state with a stabilizer group $S$. First, for $\alpha\in M$, if the subgroup $S_\alpha$ is nonzero and generated by $s_\alpha$ $Z$-type generators and $s_\alpha'$ $X$-type generators, from Lemma \ref{lemma: main lemma}, we see that $|0\rangle_\alpha^{\otimes s_\alpha} \otimes |+\rangle_\alpha^{\otimes s_\alpha'}$ can be extracted from $|\psi\rangle$ by CSS-LCU unitaries. After doing such extraction for every $\alpha\in M$, the rest state that we still call $|\psi\rangle$ will have full local ranks. 

Next, we consider extracting the following $m=|M|$-qubit GHZ states from $|\psi\rangle$, i.e.,
\beq
|\mathrm{GHZ}\rangle_M &&= \frac{1}{\sqrt{2}}\left(|00\dots 0\rangle + |11\dots 1\rangle\right)\\
|\mathrm{GHZ}'\rangle_M &&= \frac{1}{\sqrt{2}}\left(|++\dots +\rangle + |--\dots -\rangle\right)
\eeq
In order to count the number of $|\mathrm{GHZ}\rangle$ and $|\mathrm{GHZ}'\rangle$ that can be extracted from $|\psi\rangle$, we define the following two subgroups which are generated by their corresponding co-local subgroups,
\beq
S^X_{loc} = \sum_{\alpha\in M} S^X_{\hat \alpha} \quad\quad \text{and} \quad\quad S^Z_{loc} = \sum_{\alpha\in M} S^Z_{\hat \alpha}~,
\eeq
where $S^X_{\hat \alpha}$ and $S^Z_{\hat \alpha}$ are pure-$X$ and pure-$Z$ generators in $S_{\hat \alpha}$, respectively. 

\begin{theorem}
Let $|\psi\rangle$ be an $M$-party sign-free stabilizer state of $n$ qubits with a stabilizer group $S$, and the number of $Z$($X$)-type geneartors is $n_Z$($n_X$). Suppose that $m = |M| \geq 3$. The maximal number of states $|\mathrm{GHZ}\rangle$ and $|\mathrm{GHZ}'\rangle$ extractable from $|\psi\rangle$ by CSS-LCU unitaries are equal to $g_M=n_X - \mathrm{dim}\left(S^X_{loc}\right)$ and $g_M'=n_Z - \mathrm{dim}\left(S^Z_{loc}\right)$, respectively.
\end{theorem}

\begin{proof}
Note that $g_M,g_M'$ are invariant under extraction of local $|0\rangle_\alpha^{\otimes s_\alpha} \otimes |+\rangle_\alpha^{\otimes s_\alpha'}$ states. Thus we can safely assume that $|\psi\rangle$ has full local ranks. 

Define 
\beq
\mc{L}_\alpha^X = \{f\in G^X_\alpha:~~~ \omega(f, g) = 0~~~\text{for all}~~g\in S^Z_{loc}\}~
\eeq
and
\beq
\mc{L}_\alpha^Z = \{f\in G^Z_\alpha:~~~ \omega(f, g) = 0~~~\text{for all}~~g\in S^X_{loc}\}~,
\eeq
where $G_\alpha^X,G_\alpha^Z$ are $X,Z$-type generators acting only on party $\alpha$. Decompose $S^X$ and $S^Z$ as $S^X_{loc}\oplus S^X_{ent}$ and $S^Z_{loc}\oplus S^Z_{ent}$ with any $S^X_{ent}$ and $S^Z_{ent}$. Define two bilinear pairings of $\mc{L}_\alpha^Z, S^X_{ent}$ and $\mc{L}_\alpha^X, S^Z_{ent}$
\beq
\eta_\alpha: \mc{L}_\alpha^{Z/X}\otimes S^{X/Z}_{ent} \rightarrow \{0,1\},\quad \quad\quad \eta_\alpha(f,g) = \omega(f_\alpha, g_\alpha)
\eeq
We can prove that the two bilinear pairings are all non-degenerate, and in particular we have $\mathrm{dim}\left(\mc{L}_\alpha^{Z/X}\right) = \mathrm{dim}\left(S^{X/Z}_{ent}\right)$. 

To illustrate the usefulness of the above definition, recall that $|\mathrm{GHZ}\rangle$ is stablized by $\bar{\mathcal{f}}^X = \otimes_{\alpha\in M}X_\alpha$ and $\mathcal{f}_{\alpha\beta}^Z = Z_\alpha\otimes Z_\beta$. Then in this case $S^X_{ent}$ is generated by $\bar{\mathcal{f}}^X = \otimes_{\alpha\in M}X_\alpha$ and $\mc{L}_\alpha^Z$ are nothing but local $Z_\alpha$ operators. The same story goes for $|\mathrm{GHZ}'\rangle$. 

Note that $S_{ent}^X$ has dimension $g_M$ and $S_{ent}^Z$ has dimension $g_M'$. Choose an arbitrary basis $\bar{f}^X_i, i=1,\dots,g_M$ for the subspace $S^X_{ent}$ and $\bar{f}^Z_j, j=1,\dots,g_M'$ for the subspace $S^Z_{ent}$. For each $\alpha\in M$, choose the dual basis $f^Z_{\alpha i}, i=1,\dots,g_M$ for $\mc{L}_\alpha^{Z}$ and $f^X_{\alpha j}, j=1,\dots,g_M'$ for $\mc{L}_\alpha^{X}$, with respect to $\eta_\alpha$. 

Following Lemma \ref{lemma: main lemma}, We wish to extract $|\mathrm{GHZ}\rangle_M^{\otimes g_M}\otimes |\mathrm{GHZ}'\rangle_M^{\otimes g_M'}$ from $|\psi\rangle$ by constructing an appropriate linear injective map $T$ from its stabilizer group $S_{GHZ}$ to $S$, i.e., $T:S_{GHZ}\rightarrow S$. Choose the canonical (overcomplete) basis for $S_{GHZ}$ to be generated by $\bar{\mathcal{f}}^X_i= \otimes_{\alpha\in M}X_{\alpha i}$, $\mathcal{f}^Z_{\alpha\beta i} = Z_{\alpha i}\otimes Z_{\beta i}$, $i=1,\dots,g_M$ and $\bar{\mathcal{f}}^Z_j= \otimes_{\alpha\in M}Z_{\alpha j}$, $\mathcal{f}^X_{\alpha\beta j} = X_{\alpha j}\otimes X_{\beta j}$, $j=1,\dots,g_M'$. The action of $T$ on the generators should be
\beq
&& T\left(\bar{\mathcal{f}}_i^X\right)\rightarrow \bar{f}_i^X,
\quad\quad
T\left(\mathcal{f}_{\alpha \beta i}^Z\right)\rightarrow f_{\alpha i}^Z + f_{\beta i}^Z,
\quad\quad i=1,\dots,g_M\nonumber\\
&& T\left(\bar{\mathcal{f}}_j^Z\right)\rightarrow \bar{f}_j^Z, 
\quad\quad T\left(\mathcal{f}_{\alpha \beta j}^X\right)\rightarrow f_{\alpha i}^X + f_{\beta j}^X,\quad\quad j=1,\dots,g_M'~.
\eeq
For $T$ to be well-defined and to satisfy the conditions in Lemma~\ref{lemma: main lemma}, we need to prove  the following facts for any $i,j,\alpha,\beta$:
\begin{enumerate}
\item $f_{\alpha i}^Z + f_{\beta i}^Z, f_{\alpha j}^X + f_{\beta j}^X \in S$. This can be proven by noting that $f_{\alpha i}^Z + f_{\beta i}^Z, f_{\alpha j}^X + f_{\beta j}^X$ commute with all elements in $S$.
\item $\omega\left(f_{\alpha i}^Z, f_{\alpha j}^X\right) = 0$. Hence we have $\omega\left(f_{\alpha i}^Z + f_{\beta i}^Z, f_{\alpha j}^X + f_{\beta j}^X\right) = 0$.  
\item $\omega\left(\left(\bar{f}^X_i\right)_\alpha,\left(\bar{f}^Z_j\right)_\alpha\right) = 0$. This can be achieved by adjusting $S^X_{ent},S^Z_{ent}$ through adding to $\bar{f}^X_i, \bar{f}^Z_i$ elements in $S^X_{loc},S^Z_{loc}$ respectively.
\item  $\omega\left(f^Z_{\alpha i} , \bar{f}^X_{i'}\right) = \delta_{ii'}$, $\omega\left(f^X_{\alpha j} , \bar{f}^Z_{j'}\right) = \delta_{jj'}$. This is the definition of the dual basis.
\end{enumerate}
Then we see that all conditions in Lemma~\ref{lemma: main lemma} have been satisfied and therefore $|\mathrm{GHZ}\rangle_M^{\otimes g_M}\otimes |\mathrm{GHZ}'\rangle_M^{\otimes g_M'}$ can be extracted from $|\psi\rangle$. Moreover, the condition 2 and 3 in Lemma~\ref{lemma: main lemma} suggest that $g_M$ and $g_M'$ are the maximal number of $|\mathrm{GHZ}\rangle_M$ and $|\mathrm{GHZ}'\rangle_M$ that can be extracted from $|\psi\rangle$.

\end{proof}


After the extraction, the same proof of Theorem 5 in \cite{Bravyi2006} together with Lemma \ref{lemma: main lemma} can be applied to show that the rest is only EPR pairs. 
\begin{theorem}
Let $|\psi\rangle$ be a sign-free stabilizer state with full local ranks shared by three parties $M = \{A, B, C\}$ with stabilizer group $S$. After extracting tripartite GHZ states $|\mathrm{GHZ}\rangle_M^{\otimes g_M}\otimes |\mathrm{GHZ}'\rangle_M^{\otimes g_M'}$, the resulting state is CSS-LCU-equivalent to $  \ket{\mathrm{EPR}}_{AB}^{\otimes e_{AB}} \otimes
    \ket{\mathrm{EPR}}_{BC}^{\otimes e_{BC}} \otimes
    \ket{\mathrm{EPR}}_{CA}^{\otimes e_{CA}}$.
\end{theorem}

\section{Alternative derivation of MIE bound for sign-free stabilizers}\label{app:alternative_derivation}

\begin{theorem}
For a sign-free stabilizer state (CSS code), $\MIE(A:B) \leq \MI(A:B)$.
\end{theorem}

\begin{proof} 
As shown in Ref.~\cite{fattal}, given a bi-partition of the stabilizer state $|\psi\rangle$ into $A$ and $B$, it is possible to choose some generators of stabilizers $S_{AB}$ so that they are split into three groups (by Gaussian elimination and back-substitution of the corresponding binary check-matrix):
\begin{enumerate}
    \item $G_{A}^{(i)}\otimes I_B$ that acts exclusively on $A$, where $I_B$ acts on $B$ as identity.
    \item $I_{A}\otimes G_{B}^{(j)}$ that acts exclusively on $B$, where $I_A$ acts on $A$ as identity.
    \item $n$ pairs of generators $H_{A}^{(k)}\otimes H_{B}^{(k)}$ and $\bar{H}_{A}^{(k)}\otimes \bar{H}_{B}^{(k)}$, where $H_{A}^{(k)}$ and $H_{B}^{(k)}$ anticommute with $\bar{H}_{A}^{(k)}$ and $\bar{H}_{B}^{(k)}$, respectively, but commute with all the other generators of $S_{AB}$, including $H_{A}^{(k)}$ and $H_{B}^{(k)}$ with a different $k$ corresponding to other pairs. 
\end{enumerate}

Note that since generators of $S$ of the sign-free stabilizer state can be chosen to consist of pure-$X$ Pauli strings and pure-$Z$ Pauli strings with exactly the same algorithm (namely the Gaussian elimination and back-substitution of the corresponding binary check-matrix), we assume the generators in the groups above are also either pure-$X$ Pauli strings or pure-$Z$ Pauli strings. 
Now in the third group, for every pure-$X$ stabilizer generator, its corresponding partner of the pair has to be a pure-$Z$ stabilizer so that the anti-commutation property can be satisfied.
We denote the pure-$X$ stabilizer and pure-$Z$ stabilizer of pair $k$ in the third group by $H_{X,A}^{(k)}\otimes H_{X,B}^{(k)}$ and $H_{Z,A}^{(k)}\otimes H_{Z,B}^{(k)}$, respectively, and the number of pairs $n$ is equal to the entanglement entropy of $A$ or $B$, i.e., $n = E_{A} = E_B$. 

Now consider partitioning the stabilizer state $|\psi\rangle$ into three parts A, B and C. 
Again, the generators of $S_{ABC}$ can be split into three groups, $G_{AB}^{(i)}\otimes I_C$, $I_{AB}\otimes G_{C}^{(j)}$ and $p$ pairs of $H_{X,AB}^{(k)}\otimes H_{X,C}^{(k)}$ and $H_{Z,AB}^{(k)}\otimes H_{Z,C}^{(k)}$, where $p= E_{AB}$.
After measuring $C$ in the $Z$-basis, $AB$ disentangles with $C$ and the resulting state $|\psi_{AB}(C)\rangle$ 
will be stabilized by the generators $G_{AB}^{(i)}$ and $H_{Z,AB}^{(k)}$. 
The generators of this stabilizer again can be chosen to consist of three groups, $\tilde G_{A}^{(\tilde i)}\otimes I_B$, $I_{A}\otimes \tilde G_{B}^{(\tilde j)}$ and $q$ pairs $\tilde H_{X,A}^{(\tilde k)}\otimes \tilde H_{X,B}^{(\tilde k)}$ and $\tilde H_{Z,A}^{(\tilde k)}\otimes \tilde H_{Z,B}^{(\tilde k)}$, where $q=\MIE(A,B)$ and also equals $E_A$ or $E_B$ after measurement. 

Obviously, $E_A\geq q$ and $E_B\geq q$. 
If $q\geq p$, then immediately we have $\MI(A,B) = E_A + E_B - E_{AB}\geq 2q-p \geq q = \MIE(A,B)$. 
If $p > q$, then at least $p-q$ number of $H_{Z,AB}^{(k)}$ act purely on $A$ or $B$. 
Suppose one of such $H_{Z,AB}^{(k)}$ acts purely on $A$. 
Then since $H_{X,AB}^{(k)}$ anticommutes with $H_{Z,AB}^{(k)}$, the restriction of $H_{X,AB}^{(k)}$ to $A$ cannot be written as linear combination of elements in $\tilde G_{A}^{(i)}$ and $\tilde H_{Z,A}^{(\tilde k)}$ which all commute with $H_{Z,AB}^{(k)}$. 
Moreover, consider all $H_{X,AB}^{(k)}$ corresponding to $H_{Z,AB}^{(k)}$ that acts purely on $A$. For exactly the same reason, any such element, when restricted to $A$, cannot be written as linear combination of other restriction of $H_{X,AB}^{(k)}$ to $A$, i.e., $\tilde G_{A}^{(i)}$ and $\tilde H_{Z,A}^{(\tilde k)}$. Hence, $E_A+E_B$ after measurement should be decreased by at least $p-q$, i.e. $E_A + E_B - 2q \geq p-q$. 
Therefore, we again establish that $\MI(A,B)\geq \MIE(A:B)$.
\end{proof}

\section{MIE in critical three-state Potts model}
\label{app:threestatePotts}
Here we consider the three state Potts model, which has local Hilbert space dimension $3$ and is thus beyond the scope of our bound in the main text.
The three-state Potts model is defined by 
\begin{equation}
    H_{\text{Potts}} = -\sum_{j} (U_{j} U^{\dagger}_{j+1}+U^{\dagger}_j U_{j+1}) - \sum_{j} (V_j+V^{\dagger}_j),
\end{equation}
where the operators $U$ and $V$ are the analog of spin operators $Z$ and $X$ in the Ising model, 
\begin{eqnarray}
    U &=& \sum_{n=1}^{3} e^{in\frac{2\pi}{3}} |n\rangle\langle n| \\
    V &=& \sum_{n=1}^{3} |n+1\rangle\langle n|~.
\end{eqnarray}
The model has a global $\mathbb{Z}_3$ symmetry generated by $\prod_{j} V_j$. There are two possible choices of measurement basis as the local Hilbert space dimension is $3$. The first choice, which we denote by $U$, is the measurement in the computational basis,
\begin{equation}
    P_n = |n\rangle \langle n|. ~~~ (\text{basis U}).
\end{equation}
The second choice is measurement in the eigenbasis of $V$, that is
\begin{equation}
    |\tilde{k}\rangle = \sum_{n=1}^{3} e^{ink \frac{2\pi}{3}}|n\rangle,
\end{equation}
and the projectors are
\begin{equation}
    P_k = |\tilde{k}\rangle\langle \tilde{k}|, ~~~(\text{basis V}).
\end{equation}
As we numerically observe in Figure \ref{fig:threePotts}, 
the MIEs in the basis $U$ and $V$ decay polynomially with the distance, similar to the Ising and the tri-critical Ising models. The MIE in $V$ basis decays as $x^{-0.27}$, where the exponent is roughly half of the MI exponent, and the MIE in $U$ basis decays as $x^{-3.28}$ with a much larger exponent.  Again, the MIE in $V$ basis appears to track the slowest decaying correlation function $\la U^{\dagger}(x) U(0) \ra\sim x^{-4/15}$ \cite{belavin_infinite_1984}, with the exponent $4/15\approx 0.267$.
\begin{figure}
    \centering
    \includegraphics[width=0.6\columnwidth]{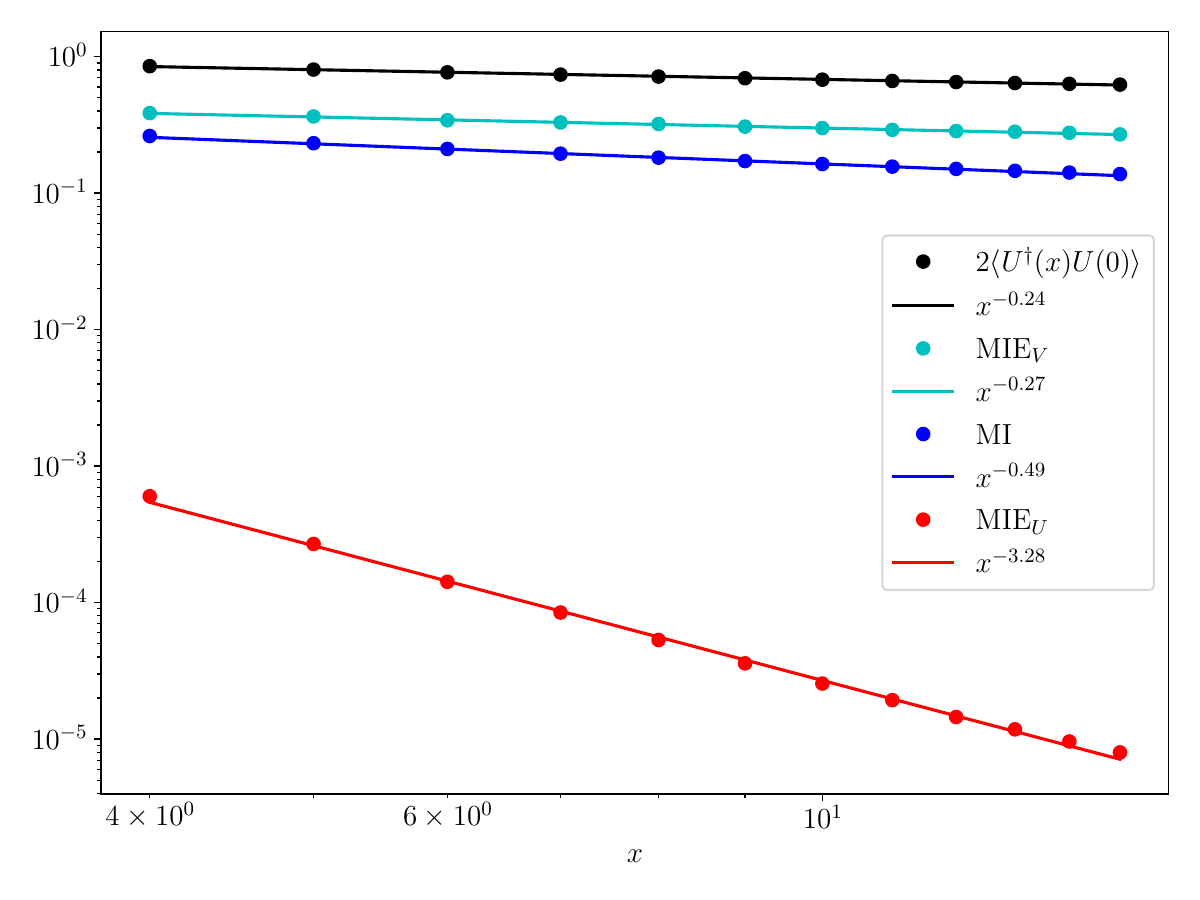}
    \caption{MI, MIE and the relevant correlations in the three-state Potts model $H_{\text{Potts}}$ with the system size $L=48$.}
    \label{fig:threePotts}
\end{figure}

\section{Sampling MIE using Matrix Product States}\label{app:sampling}

In order to reduce the sampling error, we both sample over different time slices of the steady states and different measurement outcomes. 
For Haar-MIPT, denote the number of steady states sampled be $N$ and the number of measurement outcomes sampled be $M$, then the total number of samples is $NM$. 
In the actual numerical simulation, we choose $N=200$ and $M=40000$, resulting in $8$ million samples. 
Thus the statistical error is on the order of $1/\sqrt{NM}$, which is less than $10^{-3}$.

In addition, we improve the algorithm that samples measurement outcomes to reduce the auto-correlations significantly. Previous algorithm \cite{Popp2005} uses Markov chain Monte Carlo (MCMC) over the collection of measurement outcomes.
Here we eliminate the auto-correlations completely by performing measurement site by site. 
Since measurements on different sites commute with each other, the order of measurements does not affect the probability distribution of the measurement outcomes. 
We start with a MPS in the left canonical form and measure the rightmost site, then a projector is applied according to the Born's rule. 
We then shift the canonical center to the site that is left to the rightmost site and apply the measurement to that site. 
The procedure is repeated until all spins, except the spins in $A$ and $B$ which we wish to compute the MIE, are measured. 
The MIE is simply the entanglement between $A$ and $B$ in the final state and thus we obtain a sample of MIE. 
We then start over and the sampling process is applied to the same state for $M$ times. One of the merits of the the sampling process over MCMC is that it does not have any auto-correlation between samples (that is, a new sample does not have any memory of the past samples). 
Furthermore, the MPS is kept in canonical form after each measurement, which significantly reduces the cost of computing the Born probabilities.

\section{The probability distribution of measurement induced entanglement}
\label{app:distribution}
\begin{figure}
    \centering
    \includegraphics[width=0.6\columnwidth]{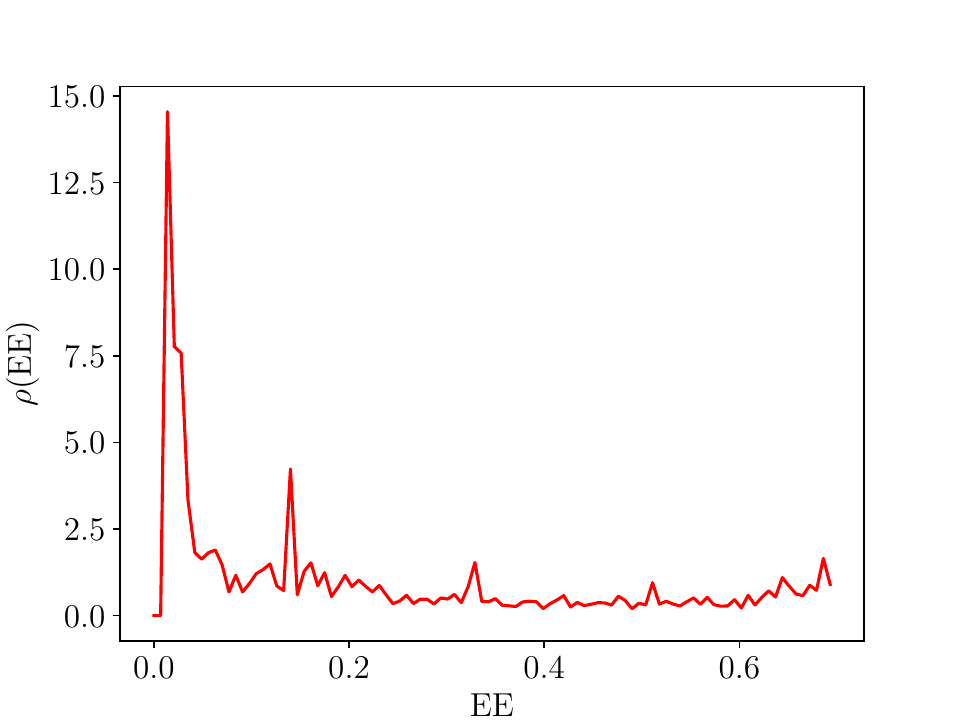}
    \caption{Probability density of entanglement entropies (EE) after measuring all but two qubits $A,B$ with separation $x=10$ in $Z$ basis for the critical transverse field Ising model $H_{\mathrm{Ising}}$ with system size $L=20$. Note that MIE (the average entanglement entropy over all measurement outcomes) is $\MIE = \int \rho(\mathrm{EE})~d\mathrm{EE} = 0.21$ here.}
    \label{fig:distribution}
\end{figure}

In this appendix, we show in Fig.~\ref{fig:distribution} the probability density of entanglement entropies between qubits $A,B$ with separation $x=10$ after measuring the complement in the $Z$-basis.  We use the ground state of the critical transverse field Ising with $L=20$ for illustration.
As we see from the distribution, the large-valued measurement induced entanglement entropies have smaller probabilities to occur, while the small-valued ones have higher probabilities.
This suggests that the behavior of MIE we observe in the main text is an averaged behavior and is not dominated by the behaviors from few special measurement outcomes.

\section{MIE, MI, and Correlations of 1d gapless SPT}\label{app:gSPTresults}
In this appendix, we give the derivation for the various results of the MI and MIE scaling in the model $H_{gSPT}$.
For reader's convenience and to be self-contained, we repeat the construction of the ground state wavefunction shown in Ref.~\cite{scaffidi}.
The Hamiltonian $H_{gSPT}$ can be related to the Hamiltonian
\beq
    H_{\text{trivial}} =  
    -\sum_j(\sigma^z_j+\tau^z_{j-\frac{1}{2}})-\sum_j \sigma^x_j\sigma^x_{j-1}~,
\eeq
by $H_{gSPT} = U H_{\text{trivial}} U^{\dagger}$, where
\beq
    U=\prod_{j=1}^{N}\text{CX}_{j,j+\frac{1}{2}}\text{CX}_{j+\frac{1}{2},j+1}~
\eeq
and $\text{CX}_{ij}$ is the controlled-$X$ gate operate on the spin pair $i$ and $j$, where the control is based on the $X$-basis of the controlled qubit.
In other words,
\beq
 \text{CX}_{ij}=|+\ra\la +|_i X_j + |-\ra\la -|_i~.
\eeq
The ground state of $H_{gSPT}$ can therefore be constructed as $|\psi\ra = U (|\psi_{\text{Ising}}\ra_{\sigma} \otimes |\pmb{1}\ra_{\tau} )$, where $|\psi_{\text{Ising}}\ra_{\sigma}$ is the critical ground state of the Ising model on the $\sigma$ spins and $|\pmb{1}\ra_{\tau} = \otimes_{j}|1\ra_{j+\frac{1}{2}}$ on the $\tau$ spins.

Some of the correlation functions and the mutual information can therefore be related to the ones in the Ising model.
In particular, since $X_j$ commutes with $U$, we have $\la \psi| \sigma_A^x \sigma_B^x |\psi \ra = \la  X_A X_B \ra_{\text{Ising}} \sim \ell^{-1/4}$ when $A, B$ are on $\sigma$, and $\la \psi| X_A X_B |\psi \ra = 0$ when any of $A$ or $B$ is on $\tau$. 
We can also obtain $\la \psi|\tau^z_A \tau^z_B |\psi\ra$ as the following. 
Since $\sigma^x_{j}\tau^z_{j+\frac{1}{2}}\sigma^x_{j+1}=1$ for the ground state, we have $\la \psi|\tau^z_A \tau^z_B |\psi\ra = \la \psi| \sigma^x_{A-\frac{1}{2}}\sigma^x_{A+\frac{1}{2}}\sigma^x_{B-\frac{1}{2}}\sigma^x_{B+\frac{1}{2}}|\psi \ra \sim \la Z_A Z_B\ra_{\text{Ising}} \sim \ell^{-2}$,
where we have used the self-dual property in the quantum Ising model.

The scaling of the mutual information can also be related to the above correlation functions.
For A and B both on $\sigma$, we have the mutual information $\text{MI} \sim \la  X_A X_B \ra_{\text{Ising}}^2 \sim \ell^{-1/2}$; while if A and B are both on $\tau$, $\text{MI} \sim \la  Z_A Z_B \ra_{\text{Ising}}^2 \sim \ell^{-4}$. 
Finally, $\text{MI} =0 $ if $A$ and $B$ are of different type of the spins. 
We derive these results in details in the following.

The density matrix is $\rho = U (\Omega \otimes |\pmb{+}\ra\la\pmb{+}|) U^{\dagger}$, where $\Omega = |\psi_{\text{Ising}}\ra \la \psi_{\text{Ising}}| $ is the density matrix of the Ising ground state on the $\sigma$ spins.
First, consider both A and B on the $\sigma$ spins. 
The reduced density matrix on A(B) is 
\beq
    \rho_{A(B)}=\frac{1}{2}(\Omega_{A(B)} + X_{A(B)}\Omega_{A(B)}X_{A(B)})~,
\eeq
where $\Omega_{I}=\tr_{I^{c}}[\Omega]$ is the reduced density matrix of $\Omega$ on the region $I$.
The reduced density matrix on $A$ or on $B$ is already diagonal in the $x$-basis, it is easy to obtain that $S_A = S_B = \ln 2$ from the Ising symmetry.
On the other hand, the reduced density matrix on $A \cup B$ can be obtained similarly
\beq
    \rho_{AB}&=&\frac{1}{4}(\Omega_{AB} + X_{A}\Omega_{AB}X_{A}+X_{B}\Omega_{AB}X_{B}
    +X_{A}X_{B}\Omega_{AB}X_{A}X_{B})~,
\eeq
which is again already diagonal in the $x$-basis.
Its eigenvalues are therefore $\frac{1}{4}\la (1+sX_A)(1+rX_B)  \ra_{\text{Ising}}$ with $s$ and $r=\pm$. 
Accordingly, the entanglement entropy is
\beq
    S_{AB}&=&\ln 4 - \frac{1}{4}\sum_{s,r=\pm}\la(1+sX_A)(1+rX_B)\ra_{\text{Ising}} \ln [\la(1+sX_A)(1+rX_B)\ra_{\text{Ising}}] \\ \notag
    &\approx& \ln 4 -\frac{1}{2}[(1+\la X_AX_B\ra_{\text{Ising}})\la X_A X_B \ra_{\text{Ising}} - (1-\la X_AX_B\ra_{\text{Ising}})\la X_A X_B \ra_{\text{Ising}}] = \ln 4 - \la X_A X_B\ra^2_{\text{Ising}}~,
\eeq
where we have assumed large separation of $A,B$ and used $\ln (1+x) \approx x$ when $|x|\ll 1$.
We therefore have $\MI \sim \la X_A X_B\ra^2_{\text{Ising}} \sim \ell^{-1/2}$ if $A$ and $B$ are both on $\sigma$.
On the other hand, if $A$ and $B$ are both on $\tau$, we have 
\beq
    \rho_A = \frac{1}{2}[|-\ra \la -|+\la \mu_{A} \ra_{\text{Ising}} (|-\ra \la +|+|+\ra\la-|) +|+\ra\la+|]~,
\eeq
where we have abbreviated $\mu_A \equiv X_{A-\frac{1}{2}}X_{A+\frac{1}{2}}$.
One can therefore easily get 
\beq
    S_A = \ln 2 - \sum_{s=\pm}\frac{1}{2}[(1+s\la \mu_A \ra_{\text{Ising}})\ln (1+s\la \mu_A \ra_{\text{Ising}})]~,
\eeq
and similarly for $\rho_B$ and $S_B$.
We can also obtain $\rho_{AB}$ similarly when $A$ and $B$ are both on $\tau$ and far apart in the matrix form as 
\beq
    \rho_{AB}&=&\frac{1}{4}
    \begin{pmatrix}
    1 & \la \mu_A \ra & \la \mu_B \ra & \la \mu_A \mu_B\ra \\
    \la \mu_A \ra & 1 & \la \mu_A \mu_B\ra & \la \mu_B \ra \\
    \la \mu_B \ra & \la \mu_A \mu_B\ra & 1 & \la \mu_A \ra \\
    \la \mu_A \mu_B\ra & \la \mu_B \ra & \la \mu_A \ra & 1
    \end{pmatrix}~,
\eeq
where all of the above expectation values are evaluated with respect to the $|\psi_{\text{Ising}}\ra$ state.
It is easy to check that the four eigenvalues of $\rho_{AB}$ are $\frac{1}{4}(1+ s\la \mu_{A} \ra +r \la \mu_{B} \ra + rs \la \mu_A \mu_B\ra)$ where $s$, $r=\pm$.
Consider the long-distance limit $\la \mu_A \mu_B\ra \sim c\ell^{-2} + \la \mu_{A} \ra\la \mu_{B} \ra $, it is easy to check that $\MI \rightarrow 0$ when $\ell \rightarrow \infty$.
The leading scaling behavior of $\MI$ can then be obtained by expanding $\MI$ order by order in $c\ell^{-2}$, and we obtain $\MI \sim \ell^{-4}$.
Finally, when $A$ is on $\sigma$ and $B$ is on $\tau$, we have $\rho_{AB}=\rho_{A,\text{Ising}}\otimes \frac{1}{2} I_B$ where $\rho_{A,\text{Ising}}=\tr_{A^c}[|\psi_{\text{Ising}}\ra\la\psi_{\text{Ising}} |]$ and $I_B$ is the identity matrix on site B. 
We therefore have $\MI(A:B)=0$ when A and B are on different type of the spins.

Some of the MIE's can also be related to the ones obtained in the Ising model. 
In particular, since the measurement in the $X$-basis commute with $U$, the results of $\MIE_X$ given in the main text can be obtained easily related to the MIE's obtained in the Ising model. 
To obtain $\MIE_Z$, let us consider running the circuit in ``backwards''. 
In particular, since we are interested in the quantity $\psi(\{a_i,b_j\})\equiv \la \{\sigma_i^z = a_i \},\{\tau_{j}^z = b_j\}|U |\psi_{\text{Ising}}\ra\otimes |\pmb{1}\ra $, where $i$ or $j$ runs through all the sites (including A and B), we consider $U |\{\sigma_i^x = a_i \},\{\tau_{j}^x = b_j\} \ra$ first.
This state is a stabilizer state stabilized by the generators $(-1)^{a_i}\tau^x_{i-\frac{1}{2}}\sigma^z_i\tau^x_{i+\frac{1}{2}}$ and $(-1)^{b_j}\sigma^x_{j-\frac{1}{2}}\tau^z_j\sigma^x_{j+\frac{1}{2}}$.
To calculating $\psi(\{a_i,b_j\})$, let us first consider it overlapping with the $\la \pmb{1}|$ state on the $\tau$ spins.
This can in fact be viewed as measuring $U |\{\sigma_i^x = a_i \},\{\tau_{j}^x = b_j\} \ra$ in the $Z$-basis on the $\tau$ spins with the $|1\ra$ outcomes.
Such a state can again be described by a stabilizer state stabilized by the generators $\{\tau^z_j, (-1)^{b_j}\sigma^x_{j-\frac{1}{2}}\sigma^x_{j+\frac{1}{2}},(-1)^{\sum_i a_i}\prod_i\sigma^z_i  \}$, where again $i$ and $j$ run through all the sites including $A$ and $B$.
We therefore see that this resulting state is a GHZ-type state on the $\sigma$ spins.
To obtain $\psi(\{a_i,b_j\})$, we then consider its overlap with $|\psi_{\text{Ising}}\ra$. 
Note that, since $|\psi_{\text{Ising}}\ra$ is in the $Z_2$ parity plus sector (recall $U_{Z_2} \equiv \prod_j\sigma^z_j$), only outcomes with $\sum_i a_i$ being even will have nonzero overlap.
Furthermore, we see that once we consider only the outcome being in the $\sum_i a_i$ even sector, the resulting wavefunction will only depend on $\{ b_j\}$ and will be independent of $\{a_i\}$.

Now considering $A$ and $B$ are both on $\sigma$, the probability of getting an outcome $C$ is $P_C = \sum_{a_A,a_B =0,1}|\psi(\{a_i,b_j\})|^2$ with the resulting wavefunction being $|\psi_C\ra = \sum_{a_A,a_B =0,1}\chi(\{a_i,b_j\})|a_A,a_B\ra$, where $\chi(\{a_i,b_j\})=\psi(\{a_i,b_j\})/\sqrt{P_C}$.
Now, since we know $\sum_i a_i$ is even, if $\sum_{i\neq A,B} a_i$ is even, then the resulting wavefunction can only have nonzero amplitudes on $(a_A,a_B)=(1,1)$ or $(a_A,a_B)=(0,0)$; if $\sum_{i\neq A,B} a_i$ is odd, then the resulting wavefunction can only have nonzero amplitudes on $(a_A,a_B)=(1,0)$ or $(a_A,a_B)=(0,1)$.
As mentioned earlier, the wavefunction $\la \pmb{+}| U |\{\sigma^x=a_i,\tau^x=b_j\}\ra$ is independent of $a_j$ if given $\{\tau_j^x=b_j\}$.
We therefore conclude that the resulting wavefunctions after measurements are either $\frac{1}{\sqrt{2}}(|11\ra+|00\ra)$ or $\frac{1}{\sqrt{2}}(|10\ra+|01\ra)$, which gives us the result $\MIE_Z=\ln 2$ if $A$ and $B$ are both on $\sigma$, independent of the separation between $A,B$. 

On the other hand, if $A$ and $B$ are both on $\tau$, we have the outcome probability $P_C=\sum_{b_A,b_B =0,1}|\psi(\{a_i,b_j\})|^2$ with the resulting wavefunction $|\psi_C\ra = \sum_{b_A,b_B =0,1}\chi(\{a_i,b_j\})|b_A,b_B\ra$.
The amplitudes will generally depending on the outcomes $\{b_j\}$. 
However, note that since we can trade $\tau_j^z$ with $\sigma_{j-\frac{1}{2}}^{x}\sigma_{j+\frac{1}{2}}^{x}$, the above resulting wavefunction will be the same as if we measure $|\psi_{\text{Ising}}\ra$ with $\sigma_{j-\frac{1}{2}}^{x}\sigma_{j+\frac{1}{2}}^{x}$ but leaving out $j=A$ and $j=B$. 
From the self-duality of $|\psi_{\text{Ising}}\ra$, this is the same as measuring $|\psi_{\text{Ising}}\ra$ on the $Z$-basis. 
We therefore conclude that $\MIE_Z(A:B)$ will be the same as $\MIE_Z(A:B)$ in the Ising model, which is consistent with our numerical results shown in Figs.~\ref{fig:MIE_gSPT}(a) and (b) in the main text.



\end{document}